\pdfoutput=1
\documentclass[acmsmall]{acmart}
\acmJournal{PACMPL}
\acmVolume{1}
\acmNumber{CONF} 
\acmArticle{1}
\acmYear{2018}
\acmMonth{1}
\acmDOI{} 
\startPage{1}

\setcopyright{none}

\usepackage{booktabs}   
\usepackage{subcaption} 
\usepackage{amsfonts}
\usepackage{amsmath}
\usepackage{array,makecell}
\usepackage{multirow}
\usepackage{enumerate,paralist}
\usepackage{listing}
\usepackage[ruled,linesnumbered]{algorithm2e}
\usepackage{xcolor}
\usepackage{listings}
\usepackage{graphicx}

\lstdefinelanguage{prog}
{
morekeywords={prob, if, then, else, fi, while, do, od, true, false, and, or, skip},
sensitive = false
}

\newcommand{\Rset}{\mathbb{R}}
\newcommand{\Nset}{\mathbb{N}}
\newcommand{\Zset}{\mathbb{Z}}

\newcommand{\val}[1]{\mbox{\sl Val}_{#1}}

\newcommand{\sampdpd}{\overline{\Upsilon}}

\newcommand{\probm}{\Pr}
\newcommand{\expv}{\mathbb{E}}
\newcommand{\condexpv}[2]{{\expv}{\left({#1}{\mid}{#2}\right)}}

\newcommand{\loc}{\ell}

\newcommand{\supp}[1]{{\mathrm{supp}}{\left(#1\right)}}

\newtheorem{remark}{Remark}

\newcommand{\rd}[1]{{#1}}

\newcommand{\yican}[1]{ #1}
\renewcommand{\time}{\mathcal{T}}

\renewcommand{\paragraph}[1]{\smallskip\noindent\textbf{\emph{#1}}.}
\newcommand{\problem}[1]{{\textsc{#1}}}
\usepackage[symbol]{footmisc}

\begin{document}
\title[Concentration-Bounds for Probabilistic Programs and Recurrences]{Concentration-Bound Analysis for Probabilistic Programs and Probabilistic Recurrence Relations}



\author{Jinyi Wang}
\authornote{equal contribution}
\affiliation{
  \institution{Shanghai Jiao Tong University}            
  \city{Shanghai}
  \country{China}
}
\email{jinyi.wang@sjtu.edu.cn}

\author{Yican Sun}
\affiliation{
  \institution{Peking University}            
  \city{Beijing}
  \country{China}
}
\email{sycpku@pku.edu.cn}
\authornotemark[1]

\author{Hongfei Fu}
\affiliation{
  \institution{Shanghai Jiao Tong University}            
  \city{Shanghai}
  \country{China}
}

\email{fuhf@cs.sjtu.edu.cn}

\author{Mingzhang Huang}
\affiliation{
  \institution{IST Austria (Institute of Science and Technology Austria)}            
  \city{Klosterneuburg}
  \country{Austria}
}
\email{mingzhang.huang@ist.ac.at}

\author{Amir Kafshdar Goharshady}
\affiliation{
  \institution{IST Austria}            
  \city{Klosterneuburg}
  \country{Austria}
}
\email{goharshady@gmail.com}

\author{Krishnendu Chatterjee}
\affiliation{
  \institution{IST Austria}            
  \city{Klosterneuburg}
  \country{Austria}
}
\email{krishnendu.chatterjee@ist.ac.at}

\renewcommand{\shortauthors}{Wang and Sun, et al.}


\begin{abstract}
	
	Analyzing probabilistic programs and randomized algorithms are classical 
	problems in computer science.
	The first basic problem in the analysis of stochastic processes is to consider 
	the expectation or mean, and another basic problem is to consider 
	concentration bounds, i.e.~showing that large deviations from the mean have small probability.
	Similarly, in the context of probabilistic programs and randomized algorithms, the analysis 
	of expected termination time/running time and their concentration bounds are 
	fundamental problems.
	In this work, we focus on concentration bounds for probabilistic programs and 
	probabilistic recurrences of randomized algorithms.
	For probabilistic programs, the basic technique to achieve concentration bounds is to
	consider martingales and apply the classical Azuma's inequality~\cite{azuma1967weighted}.
	For probabilistic recurrences of randomized algorithms, Karp's classical 
	``cookbook'' method~\cite{DBLP:journals/jacm/Karp94}, which is similar to the master theorem for recurrences, is the 
	standard approach to obtain concentration bounds. 
	In this work, we propose a novel approach for deriving concentration bounds 
	for probabilistic programs and probabilistic recurrence relations through the
	synthesis of exponential supermartingales. 
	For probabilistic programs, we present algorithms for synthesis of such 
	supermartingales in several cases. We also show that our approach can derive 
	better concentration bounds than simply applying the classical
	Azuma's inequality over various probabilistic programs considered in the 
	literature.
	For probabilistic recurrences, our approach can derive tighter bounds than 
	the well-established methods of \cite{DBLP:journals/jacm/Karp94} on classical algorithms such as quick sort, quick select, and randomized diameter computation. 
	Moreover, we show that our approach could derive bounds comparable to the optimal bound for quicksort, proposed in ~\cite{DBLP:journals/jal/McDiarmidH96}. 
	We also present a prototype implementation that can automatically infer 
	\yican{these} bounds. 

\end{abstract}




\maketitle

\section{Introduction}
In this work, we present new methods to obtain concentration bounds for probabilistic
programs and probabilistic recurrences.
In this section, we start with a brief description of probabilistic programs and recurrences and provide an overview of the basic analysis problems. Then, we discuss previous results and finally present our contributions.

\paragraph{Probabilistic programs and recurrences}
The formal analysis of probabilistic models is a fundamental problem that
spans across various disciplines, including 
probability theory and statistics~\cite{Durrett,Howard,Kemeny,Rabin63,PazBook},
randomized algorithms~\cite{DBLP:books/cu/MotwaniR95},
formal methods~\cite{BaierBook,prism},
artificial intelligence~\cite{LearningSurvey,kaelbling1998planning},
and programming languages~\cite{handbook}.
The analysis of probabilistic programs, i.e.~imperative programs extended with 
random value generators, has received significant attention in the programming 
languages community~\cite{SriramCAV,HolgerPOPL,SumitPLDI,EGK12,DBLP:journals/toplas/ChatterjeeFNH18,DBLP:journals/pacmpl/BartheEGHS18,DBLP:conf/popl/BartheGHS17,DBLP:conf/esop/FosterKMR016,DBLP:conf/pldi/Cusumano-Towner18a,DBLP:conf/cav/ChatterjeeFG16,DBLP:journals/pacmpl/HarkKGK20,DBLP:conf/lics/OlmedoKKM16,DBLP:journals/pacmpl/DahlqvistK20}.
Similarly, the analysis of randomized algorithms is a central and classical problem in theoretical
computer science~\cite{DBLP:books/cu/MotwaniR95}, and a key problem is to analyze the probabilistic 
recurrence relations arising from randomized algorithms.

\paragraph{Basic analysis problems}
In the analysis of probabilistic programs and probabilistic recurrences, one has to analyze the underlying stochastic processes.
The first basic problem in analysis of stochastic processes is to consider 
the expectation or the mean~\cite{probabilitycambridge}.
However, the expectation (or the first moment) does not provide enough information about
the probability distribution associated with the process.
Hence higher moments (such as variance) are key parameters of interest.
Another fundamental problem is to obtain {\em concentration bounds} showing that large deviation from the mean 
has small probability.
A key advantage of establishing strong concentration bounds is that it enables us to provide high-probability guarantees (e.g.~with high probability the running time does not exceed a desired
bound).
In the context of probabilistic programs and randomized algorithms, a key quantity of interest is
the termination time of programs or the running time associated with probabilistic recurrences.
Thus, the analysis of expected termination time/running time and their concentration bounds 
for probabilistic programs and probabilistic recurrences are fundamental problems in computer science.

\paragraph{Previous results on expectation analysis}
The problem of expected termination time analysis has received huge attention.
The expected termination time analysis for probabilistic programs has been widely 
studied with different techniques, e.g.~martingale-based techniques~\cite{SriramCAV,DBLP:journals/corr/ChatterjeeF17} and
weakest pre-expectation calculus~\cite{DBLP:conf/esop/KaminskiKMO16}.
The analysis of expected running time of probabilistic recurrences is a fundamental
problem studied in randomized algorithms~\cite{DBLP:books/cu/MotwaniR95}, and automated methods for them
have also been considered~\cite{DBLP:conf/cav/ChatterjeeFM17}.

\paragraph{Previous results on concentration bounds}
The analysis of concentration bounds is more involved than expectation analysis, both for the 
general case of stochastic processes, as well as in the context of probabilistic programs
and probabilistic recurrences.
For probabilistic programs, the only technique in the literature to obtain concentration
bounds is to consider either linear or polynomial supermartingales~\cite{SriramCAV,DBLP:conf/cav/ChatterjeeFG16}, and then apply the 
classical Azuma's inequality~\cite{azuma1967weighted} on the supermartingales to obtain concentration bounds.
For probabilistic recurrences of randomized algorithms, the standard approach for concentration 
bounds is Karp's classical ``cookbook'' method~\cite{DBLP:journals/jacm/Karp94}, which is similar to the master theorem for 
non-probabilistic recurrences.
More advanced methods have also been developed for specific algorithms such as \problem{QuickSort}~\cite{DBLP:journals/jal/McDiarmidH96,DBLP:books/daglib/0025902}.

\paragraph{Our contributions}
In this work, we consider the problem of concentration bounds for probabilistic programs and probabilistic 
recurrences and our main contributions are as follows:
\begin{enumerate}

\item First, we propose a novel approach for deriving concentration bounds 
for probabilistic programs and probabilistic recurrence relations through the
synthesis of exponential supermartingales. 

\item For probabilistic programs, we present algorithms for the synthesis of such 
supermartingales for several cases.
We show that our approach can derive better concentration bounds than simply applying 
the classical Azuma's inequality, over various probabilistic programs considered in the 
literature.

\item For probabilistic recurrences, we show that our approach can derive tighter bounds 
than the \yican{classical methods} of the literature on several basic problems.
We show that our concentration bounds for probabilistic recurrences associated with classical
randomized algorithms such as \problem{QuickSelect} (which generalizes median selection), \problem{QuickSort} and
randomized diameter computation beat the bounds obtained using methods of~\cite{DBLP:journals/jacm/Karp94}.
Moreover, we show that our approach could derive bounds comparable to the optimal bound for quicksort, proposed in ~\cite{DBLP:journals/jal/McDiarmidH96}. 
We also present a prototype implementation that can automatically infer 
these bounds. 

\end{enumerate}

\paragraph{Novelty}
The key novelty of our approach is as follows:
Instead of linear or polynomial martingales, we consider exponential martingales,
and our concentration bounds are obtained using the standard Markov's inequality.
This is quite a simple technique. 
The surprising and novel aspect is that with such a simple technique, we can  
(a)~obtain better concentration bounds for probabilistic programs in comparison with the only existing method of applying Azuma's inequality; and 
(b)~improve the more than two-decades old classical bounds for basic and well-studied randomized algorithms.

\section{Preliminaries}~\label{sect:preliminaries}

Throughout the paper, we denote by $\Nset$, $\Nset_0$, $\Zset$, and $\Rset$ the sets of positive integers, non-negative integers, integers, and real numbers, respectively.
We first review several fundamental concepts in probability theory and then illustrate the problems to study.

\subsection{Basics of Probability Theory}\label{subsect:basics}

We provide a short review of some necessary concepts in probability theory. For a more detailed treatment, see~\cite{probabilitycambridge}.

\paragraph{Probability Distributions} A \emph{discrete probability distribution} over a countable set $U$ is a function $p:U\rightarrow[0,1]$ such that $\sum_{u\in U}p(u)=1$.
The \emph{support} of $p$ is defined as $\supp{p}:=\{u\in U\mid p(u)>0\}$.

\paragraph{Probability Spaces} A \emph{probability space} is a triple $(\Omega,\mathcal{F},\probm)$, where $\Omega$ is a non-empty set (called the \emph{sample space}), $\mathcal{F}$ is a \emph{$\sigma$-algebra} over $\Omega$ (i.e.~a collection of subsets of $\Omega$ that contains the empty set $\emptyset$ and is closed under complementation and countable union) and $\probm$ is a \emph{probability measure} on $\mathcal{F}$, i.e.~a function $\probm\colon \mathcal{F}\rightarrow[0,1]$ such that (i) $\probm(\Omega)=1$ and
(ii) for all set-sequences $A_1,A_2,\dots \in \mathcal{F}$ that are pairwise-disjoint
(i.e.~$A_i \cap A_j = \emptyset$ whenever $i\ne j$)
it holds that $\sum_{i=1}^{\infty}\probm(A_i)=\probm\left(\bigcup_{i=1}^{\infty} A_i\right)$.
Elements of $\mathcal{F}$ are called \emph{events}.
An event $A\in\mathcal{F}$ holds \emph{almost-surely} (a.s.) if $\probm(A)=1$.

\paragraph{Random Variables} A \emph{random variable} $X$ from a probability space $(\Omega,\mathcal{F},\probm)$
is an $\mathcal{F}$-measurable function $X\colon \Omega \rightarrow \Rset \cup \{-\infty,+\infty\}$, i.e.~a function such that for all $d\in \Rset \cup \{-\infty,+\infty\}$, the set $\{\omega\in \Omega\mid X(\omega)<d\}$ belongs to $\mathcal{F}$.

\paragraph{Expectation} The \emph{expected value} of a random variable $X$ from a probability space $(\Omega,\mathcal{F},\probm)$, denoted by $\expv(X)$, is defined as the Lebesgue integral of $X$ w.r.t. $\probm$, i.e.~$\expv(X):=\int X\,\mathrm{d}\probm$.
The precise definition of Lebesgue integral is somewhat technical and is
omitted  here (cf.~\cite[Chapter 5]{probabilitycambridge} for a formal definition).
If $\mbox{\sl range}~X=\{d_0,d_1,\ldots\}$ is countable, then we have
$\expv(X)=\sum_{k=0}^\infty d_k\cdot \probm(X=d_k)$.


\paragraph{Filtrations}
A \emph{filtration} of a probability space $(\Omega,\mathcal{F},\probm)$ is an infinite sequence $\{\mathcal{F}_n \}_{n\in\Nset_0}$ of $\sigma$-algebras over $\Omega$ such that $\mathcal{F}_n \subseteq \mathcal{F}_{n+1} \subseteq\mathcal{F}$ for all $n\in\Nset_0$. Intuitively, a filtration models the information available at any given point of time.

\paragraph{Conditional Expectation}
Let $X$ be any random variable from a probability space $(\Omega, \mathcal{F},\probm)$ such that $\expv(|X|)<\infty$.
Then, given any $\sigma$-algebra $\mathcal{G}\subseteq\mathcal{F}$, there exists a random variable (from $(\Omega, \mathcal{F},\probm)$), denoted by $\condexpv{X}{\mathcal{G}}$, such that:
\begin{compactitem}
\item[(E1)] $\condexpv{X}{\mathcal{G}}$ is $\mathcal{G}$-measurable, and
\item[(E2)] $\expv\left(\left|\condexpv{X}{\mathcal{G}}\right|\right)<\infty$, and
\item[(E3)] for all $A\in\mathcal{G}$, we have $\int_A \condexpv{X}{\mathcal{G}}\,\mathrm{d}\probm=\int_A {X}\,\mathrm{d}\probm$.
\end{compactitem}
The random variable $\condexpv{X}{\mathcal{G}}$ is called the \emph{conditional expectation} of $X$ given $\mathcal{G}$.
The random variable $\condexpv{X}{\mathcal{G}}$ is a.s. unique in the sense that if $Y$ is another random variable satisfying (E1)--(E3), then $\probm(Y=\condexpv{X}{\mathcal{G}})=1$.
We refer to ~\cite[Chapter~9]{probabilitycambridge} for details. Intuitively, $\condexpv{X}{\mathcal{G}}$ is the expectation of $X$, when assuming the information in $\mathcal{G}$.

\paragraph{Stochastic Processes}
A \emph{(discrete-time) stochastic process} is a sequence $\Gamma=\{X_n\}_{n\in\Nset_0}$ of random variables where $X_n$'s are all from some probability space $(\Omega,\mathcal{F},\probm)$.
The process $\Gamma$ is \emph{adapted to} a filtration $\{\mathcal{F}_n\}_{n\in\Nset_0}$ if for all $n\in\Nset_0$, $X_n$ is $\mathcal{F}_n$-measurable. Intuitively, the random variable $X_i$ models some value at the $i$-th step of the process.

%



\paragraph{Martingales and Supermartingales} A stochastic process $\Gamma=\{X_n\}_{n\in\Nset_0}$ adapted to a filtration $\{\mathcal{F}_n\}_{n\in\Nset_0}$ is a \emph{martingale} (resp. \emph{supermartingale})
if for every $n\in\Nset_0$, $\expv(|X_n|)<\infty$ and it holds a.s. that $\condexpv{X_{n+1}}{\mathcal{F}_n} = X_n$ (resp. $\condexpv{X_{n+1}}{\mathcal{F}_n}\le X_n$).
We refer to ~\cite[Chapter~10]{probabilitycambridge} for a deeper treatment.

Intuitively, a martingale (resp. supermartingale) is a discrete-time stochastic process in which for an observer who has seen the values of $X_0, \ldots, X_n$, the expected value at the next step, i.e.~$\condexpv{X_{n+1}}{\mathcal{F}_n}$, is equal to (resp. no more than) the last observed value $X_n$. Also, note that in a martingale, the observed values for $X_0, \ldots, X_{n-1}$ do not matter given that $\condexpv{X_{n+1}}{\mathcal{F}_n} = X_n.$ In contrast, in a supermartingale, the only requirement is that $\condexpv{X_{n+1}}{\mathcal{F}_n} \leq X_n$ and hence $\condexpv{X_{n+1}}{\mathcal{F}_n}$ may depend on $X_0, \ldots, X_{n-1}.$ Also, note that $\mathcal{F}_n$ might contain more information than just the observations of $X_i$'s.

\paragraph{Stopping Times} Given a probability space $(\Omega,\mathcal{F},\probm)$, a \emph{stopping time} w.r.t a filtration $\{\mathcal{F}_n\}_{n\in\Nset_0}$ is a random variable $T: \Omega\rightarrow \Nset_0\cup\{\infty\}$ such that $\{\omega\mid T(\omega)=n\}\in \mathcal{F}_n$.

\rd{
\begin{example}
	Consider an unbiased and discrete random walk, in which we start at a position $X_0$, and at each second walk one step to either left or right with equal probability. Let $X_n$ denote our position after $n$ seconds. It is easy to verify that $\expv\left(X_{n+1} \vert X_0, \ldots, X_n\right) = \frac{1}{2} (X_n - 1) + \frac{1}{2} (X_n + 1) = X_n.$ Hence, this random walk is a martingale. Note that by definition, every martingale is also a supermartingale.
	As another example, consider the classical gambler's ruin: a gambler starts with $Y_0$ dollars of money and bets continuously until he loses all of his money. If the bets are unfair, i.e.~the expected value of his money after a bet is less than its expected value before the bet, then the sequence $\{Y_n\}_{n \in \mathbb{N}_0}$ is a supermartingale. In this case, $Y_n$ is the gambler's total money after $n$ bets. On the other hand, if the bets are fair, then $\{Y_n\}_{n \in \mathbb{N}_0}$ is a martingale.
\end{example}
}

\subsection{Termination and Concentration Bounds}

In this work, we consider concentration bounds of the termination time of probabilistic programs and recurrence relations. Below, we illustrate the notions of termination time and concentration bound.

\paragraph{Termination Time} The termination-time random variable counts the number of steps a program takes to termination. In the setting of probabilistic programs, the termination-time random variable is defined as the first time the program hits its termination program counter, which is a stopping time.
For convenience, we treat each termination-time random variable as the total accumulated cost until the program terminates, for which each execution step of the program causes a single unit of cost.
In the setting of probabilistic recurrence relations, since the amount of steps in a preprocessing stage is directly added to the termination time, we treat the number of steps in a preprocessing stage as the cost at current execution step, and define the termination-time variable as the total accumulated cost until the recurrence relation runs into the empty stack (i.e., terminates).
In both cases, we consider the termination time random variable as the total accumulated cost until the stopping time of termination.

\paragraph{Concentration Bounds} A \emph{concentration bound} of a random variable $C$ is an inequality of the form $\probm(C\ge d)\le f(d)$ or $\probm(C\le d)\le f(d)$ which specifies whether the probability that the value of $C$ grows too large or too small is bounded by a function $f$, which tends to zero when $d\rightarrow\infty$ (for $\probm(C\ge d)$) or $d\rightarrow-\infty$ (for $\probm(C\le d)$).
Concentration bounds are an important probabilistic property of random variables as they witness that a random variable will have a substantially small probability of large deviation.
When applied to termination time, the concentration bound specifies the following property: the probability that a randomized algorithm runs inefficiently, i.e.~takes much longer than its expected termination time, is very small. Thus compared with expected termination time, concentration bounds provide a much finer characterization of the running time of randomized algorithms.

\paragraph{Problem Statement} In this work, we consider the following problems: For probabilistic programs, we consider unnested probabilistic while loops. Our aim is to have automated approaches for deriving much tighter exponentially-decreasing concentration bounds in comparison with the existing approach through Azuma or Hoeffding inequalities~\cite{DBLP:journals/toplas/ChatterjeeFNH18}. For probabilistic recurrence relations, our goal is to derive substantially finer bounds compared with the classical results of \cite{DBLP:journals/jacm/Karp94}, and for quicksort, we try to derive bounds comparable to the optimal result proposed in \cite{DBLP:journals/jal/McDiarmidH96}.
Obtaining more precise (tighter) concentration bounds is important in many applications, e.g.~resource-contention resolution~\cite[Chapter 13]{DBLP:books/daglib/0015106}, or in resource-constrained environments such as embedded systems. 

\section{The General Method}


In this section, we establish the mathematical foundation for the analysis of concentration bounds, based on which we will later develop novel sound approaches to derive concentration bounds for the termination time arising from probabilistic programs and recurrence relations.
As in certain situations, e.g. probabilistic recurrence relations, the termination time is defined as the total cost accumulated at every execution step, we consider concentration bounds for a general non-negative random variable.
First, we introduce the basic result through which we derive our concentration bounds. This lemma is similar to the basis for Chernoff bound.

\smallskip
\begin{lemma}~\label{lem:1}
For any real number $d$, random variable $C$ and $\beta>1$, we have $\probm(C\ge d)\le \expv(\beta^C)/\beta^d$.
\end{lemma}
\begin{proof}
Since $\beta>1$, we have $\probm(C\ge d)=\probm(\beta^C\ge \beta^d)$. Then by applying Markov's inequality, we obtain the desired result.
\end{proof}

Lemma~\ref{lem:1} provides a mathematical way to derive a concentration bound for any random variable $C$, but yet we do not know an upper bound for $\expv(\beta^C)$.
In the setting of Chernoff bounds, one often chooses to evaluate $\expv(\beta^C),$ e.g.~when the random variable $C$ is a finite sum of independent random variables.
In our setting, the situation is more complicated, because we consider $C$ to be the total accumulated cost until a stopping time\footnote{If the cost accumulated at each step is equal to one, then the total accumulated cost is equal to the running time.}.
We use supermartingales and Optional Stopping Theorem~\cite[Chapter 10.10]{probabilitycambridge} to derive an upper bound for $\expv(\beta^C)$, as is shown by the following proposition.
In the proposition below, $\{C_n\}_{n\in\Nset_0}$ is a stochastic process where each $C_n$ represents the amount of cost accumulated at the $n-$th step so that $C=\sum_{n=0}^{T-1} C_n,$
and $\{X_n\}_{n\in\Nset_0}$ is a key stochastic process for deriving an upper bound for $\expv(\beta^C)$.

\smallskip
\begin{proposition}~\label{lem:2}
Consider two stochastic processes $\{X_n\}_{n\in\Nset_0}$ and  $\{C_n\}_{n\in\Nset_0}$ adapted to a filtration $\{\mathcal{F}_n\}$ for which $X_0$ is a constant random variable, and a stopping time $T$ w.r.t the filtration $\{\mathcal{F}_n\}$.
Let $\alpha,\beta>1$ be real numbers.
If it holds that
\begin{compactitem}
\item $\mathrm{(C1):}$ $T$ is a.s. finite, i.e., $\probm(T<\infty)=1$, and
\item $\mathrm{(C2):}$ $X_T\ge K$ a.s. for some constant $K\le 0$, and
\item $\mathrm{(C3):}$ $\beta^{C_n}\cdot \condexpv{\alpha^{X_{n+1}}}{\mathcal{F}_n}\le \alpha^{X_n}$ a.s. for all $n\in\Nset_0$,
\end{compactitem}
then we have $\expv(\beta^{C})\le \alpha^{X_0-K}$ where $C:=\sum_{n=0}^{T-1} C_n$.
Here, the existence and (probabilistic) uniqueness of the conditional expectation $\condexpv{\alpha^{X_{n+1}}}{\mathcal{F}_n}$ is due to its non-negativity (see \cite[Proposition~3.1]{AgrawalCN18}).
\end{proposition}
\begin{proof}
Define the stochastic process $\{Y_n\}_{n\ge 0}$ as $Y_n := \alpha^{X_{n}}\cdot \beta^{\sum_{j=0}^{n-1}C_j}$.
By definition, we have that $Y_n>0$ for all $n$. Note that although $Y_n$ may be non-integrable (due to its exponential construction), its conditional expectation $\condexpv{Y_n}{\mathcal{F}_n}$ still exists, following from $Y_n>0$ and~\cite[Proposition~3.1]{AgrawalCN18}.
By the condition (C3) and the ``take out what is known'' property of conditional expectation (see~\cite[Chapter 9.7]{probabilitycambridge}),
we have that $\condexpv{\alpha^{X_{n+1}}\cdot\beta^{\sum_{j=0}^{n} C_j}}{\mathcal{F}_n}\le \alpha^{X_n}\cdot
\beta^{\sum_{j=0}^{n-1} C_j}$.
(Note that although the integrability may not be ensured, but the proof on~\cite[Page 90]{probabilitycambridge} guarantees the case for non-negative random variables.)
Then we have that $\condexpv{Y_{n+1}}{\mathcal{F}_n} \le Y_n$ a.s., which means $\expv(Y_{n+1})\le \expv(Y_n)$ from the basic property of conditional expectation.
It follows from an easy induction on $n$ that $\expv(Y_n)\le \expv(Y_0)<\infty$ for all $n\ge 0$, thus the conditional expectation is also taken in the normal sense as each $Y_n$ is indeed integrable, and
$\{Y_n\}_{n\in\Nset_0}$ is then a supermartingale.
Moreover, the process $\{Y_n\}_{n\ge 0}$ is a non-negative supermartingale by definition.
Then by applying Optional Stopping Theorem~\cite[Chapter 10.10]{probabilitycambridge} and using (C1),
we obtain that $\expv(Y_T)\le \expv(Y_0)=\alpha^{X_0}$.
Now, from the condition (C2), we have $Y_T\ge \alpha^K \cdot \beta^C$ a.s.
It follows that $\expv(\beta^C)\le \alpha^{X_0-K}$.
\end{proof}


Proposition~\ref{lem:2} provides a way to bound  $\expv(\beta^C)$ by a stochastic process $\{X_n\}_{n\ge 0}$ and an auxiliary $\alpha>1$ if the conditions (C1)--(C3) are satisfied.
By combining Lemma~\ref{lem:1} and Proposition~\ref{lem:2}, we obtain the following main theorem of this section for concentration bounds of total accumulated cost until a stopping time.







\smallskip
\begin{theorem}~\label{thm:main}
Let $\{X_n\}_{n\in\Nset_0}$ and $\{C_n\}_{n\in\Nset_0}$ be stochastic processes adapted to a filtration $\{\mathcal{F}_n\}$ for which $X_0$ is a constant random variable, and $T$ be a stopping time  with respect to
the filtration $\{\mathcal{F}_n\}$.
Let $\alpha,\beta>1$ be real numbers.
If the conditions (C1)--(C3) (cf. Proposition~\ref{lem:2}) are fulfilled (by $X_n, C_n$'s, $\alpha,\beta$ and $T$), then we have that $\probm(C\geq  d)\leq \alpha^{X_0-K}\cdot \beta^{-d}$ for every $d\in \Rset$,
where  $C:=\sum_{n=0}^{T-1} C_n$.
\end{theorem}
\begin{proof}
By Lemma~\ref{lem:1}, we have $\probm(C\geq d)\le \expv(\beta^C)/\beta^{d}$.
By Proposition~\ref{lem:2}, we have $\expv(\beta^C)\le \alpha^{X_0-K}$.
Thus, we have $\probm(C\geq  d)\le \alpha^{X_0-K}\cdot \beta^{-d}$.
\end{proof}

\begin{remark}[Comparison with Classical Methods]
Compared with classical mathematical methods such as Chernoff bounds and Hoeffding's inequality that only examine the expected value, variance or range of related random variables, our method (Theorem~\ref{thm:main}) examines the detailed probability distribution by deriving a better bound for
the moment generation function  $\expv(\beta^C)$. The derivation is through the construction of an exponential supermartingale (see Proposition~\ref{lem:2}), and depends on the detailed probability distributions in the considered problem.
Since we examine the probability distributions in detail,
our method can obtain much tighter concentration bound than the inequalities of Azuma and Hoeffding~\cite{DBLP:journals/toplas/ChatterjeeFNH18} and Karp's classical method~\cite{DBLP:journals/jacm/Karp94}. We do not consider Chernoff bounds as they are typically used for a finite sum of independent random variables, while our method focuses on concentration bounds for total accumulated costs until a stopping time.
We show how the method can be automated for several classical and widely-used probability distributions, such as uniform and Bernoulli. 
\end{remark}

\section{Concentration Bounds for Probabilistic Programs}

In this section, we apply our method (Theorem~\ref{thm:main}) to probabilistic programs. Probabilistic programs are imperative programs extended with random number generators. Here we consider the concentration bound of a single probabilistic while loop in which there is no nested loops.
We first illustrate how we can apply our method to a single probabilistic while loop.
Then we develop automated algorithms for deriving concentration bound.
Finally, we present experimental results and compare our results with previous methods.

\subsection{The General Method Applied to Probabilistic While Loops}\label{subsec:syntax_and_semantics} 

We consider a probabilistic while loop of the form
\begin{equation}\label{eq:probprogramform}
\textbf{while}~G~\textbf{do}~Q~\textbf{od}
\end{equation}
where $G$ is the loop guard and $Q$ an imperative probabilistic program with possibly assignment statements with samplings, conditional branches, probabilistic branches, but without (nested) while loops.
For a detailed treatment of the syntax, we refer to~\cite{DBLP:journals/toplas/ChatterjeeFNH18}.
Since we only consider single probabilistic while loops here, we choose to have a light-weight description of probabilistic programs.

Before we describe probabilistic programs, we first illustrate a simple example.


\begin{example}\label{ex:simpleprog}
	Consider the probabilistic while loop:
	\lstset{language=prog}
	\begin{lstlisting}[mathescape]
	while($x\geq 0$)do $x:=x+r$ od
	\end{lstlisting}
	In each loop iteration, the value of the variable $x$ is added a random sampling $r$ whose distribution is given by $\probm(r=-1)=0.75, \probm(r=1)=0.25$.
	The random value of $r$ is sampled independently in every iteration.
	The loop ends when the value of $x$ falls below zero.
\end{example}

Below we describe the behaviours of a probabilistic while loop. In the sequel, we fix two disjoint sets of variables: the set $V_p$ of \emph{program variables} and the set $V_r$ of \emph{sampling variables}.
Informally, program variables are real-valued and are directly related to the control flow of a program, while sampling variables represent random inputs sampled from their predefined probability distributions.

\smallskip
\noindent{\em Valuations.} A \emph{valuation} over a finite set $V$ of variables is a function $\nu : V \rightarrow \Rset$ that assigns a real value to each variable.
We denote by $\val{V}$ the set of all valuations over $V$.
For simplicity, we treat each loop guard $G$ of a probabilistic while loop as a subset of $\val{V_{\mathrm{p}}}$ so that a valuation satisfies the loop guard iff the valuation falls in the designated subset $G$.
Moreover, we may also  regard each element $\nu
in \val{{V_{\mathrm{p}}}}$ (resp. $\mu\in \val{V_{\mathrm{r}}}$) as a vector $\nu$ (resp. $\mu$) with an implicit ordering between program variables (resp. sampling variables),
such that $\nu[i]$ (resp. $\mu[i]$) represents the value of the $i$th program variable (resp. sampling variable), respectively.

\smallskip
\noindent{\em The Semantics.} We present a lightweight semantics for probabilistic while loops.
We describe the execution of a probabilistic while loop by an infinite sequence of vectors of random variables $\{\nu_n\}_{n\geq 0}$
inductively defined as follows. Below we fix a probabilistic while loop in the form (\ref{eq:probprogramform}).

\begin{compactitem}
	\item  $\nu_0$ is a vector of constant random variables representing the initial input;
	\item  if $\nu\in G$ (i.e., $\nu$ satisfies the loop guard), then $\nu_{n+1} = F(\nu_n, \mu_n)$, where (i) the vector $\mu_n$ represents the sampled values for sampling variables at the $(n+1)$th loop iteration,
	and (ii) $F$ is the \emph{update function} for the loop body such that given the current valuation  $\nu\in\val{{V_{\mathrm{p}}}}$ for the program variables before the current loop iteration and the values
	$\mu\in\val{{V_{\mathrm{r}}}}$ sampled in the current loop iteration, $F(\nu, \mu)$ is the updated vector for program variables after the execution of the current loop iteration;
	\item if $\nu\not\in G$ (i.e., $\nu$ does not satisfy the loop guard), then  $\nu_{n+1} = \nu_n$.
\end{compactitem}
\noindent The inductive definition can be turned into a general state-space Markov chain by defining suitable kernel functions for the update function $F$, see~\cite[Chatper 3]{gssmc}.

In the paper, we consider the simplified setting that the sampled values are discrete with bounded range.
For this simplified setting, we use a \emph{sampling function} $\Upsilon$ to describe the sampling, which assigns to every sampling variable $r\in {V_{\mathrm{r}}}$ the predefined discrete probability distribution over $\Rset$.
Then, the \emph{joint} discrete probability distribution $\sampdpd$ over $\val{{V_{\mathrm{r}}}}$ is defined as
$\sampdpd(\mu):=\prod_{r\in{V_{\mathrm{r}}}} \Upsilon(r)(\mu(r))$ for all valuations $\mu$ over sampling variables.

Below we apply our method (Theorem~\ref{thm:main}) to probabilistic while loops, where we consider that the cost $C_n$ to execute the $(n+1)$th loop iteration is equal to $1$,
so that the total accumulated cost is equal to the number of loop iterations of the loop.
Recall that $F$ is the update function for the loop body.
\smallskip
\begin{theorem}[Concentration for Probabilistic Programs]~\label{thm:sound4}
	Let $P$ be a probabilistic while loop in the form~(\ref{eq:probprogramform}) and $T$ be the random variable representing the number of loop iterations of $P$ until termination.
	Suppose that (i) $T$ is a.s. finite (i.e., $P$ terminates with probability one) and (ii) there exist real numbers $\alpha,\beta>1, K\le 0$ and a Borel measurable function $\eta:\val{V_\mathrm{p}}\rightarrow \Rset$ such that
	\begin{compactitem}
		\item $\mathrm{(A1):}$ $\eta(\nu)\geq K$ for all $\nu\in G$;
		\item $\mathrm{(A2):}$ $\eta(\nu,\mu)\geq K$ for all $\nu\in G$ and $\mu\in V_{\mathrm{r}}$;
		\item $\mathrm{(A3):}$ $\beta\cdot \sum_{\mu\in\val{V_{\mathrm{r}}}}\sampdpd(\mu)\alpha^{\eta(F(\nu,\mu))}\leq \alpha^{\eta(\nu)}$  for every $\nu\in G$.
	\end{compactitem}
	Then for any initial valuation $\nu_0\in G$,
	we have  $\probm(T\geq n)\leq \alpha^{\eta(\nu_0)-K}\cdot\beta^{-n}$ for all $n\geq 0$.
\end{theorem}
\begin{proof}
	Choose $C_n=1$ for every $n\ge 0$. Then $T=C=\sum_{n=0}^{T-1} C_n$.
	Let $\nu_n$ be the vector of random variables representing the valuation for program variables at the $n$th step (where $\nu_0$ is a vector of constant random variables that represents the initial input).
	Define the filtration $\{\mathcal{F}_n\}_{n\ge 0}$ for which $\mathcal{F}_n$ is the smallest sigma-algebra that makes all random variables in $\nu_0, \dots, \nu_n$ measurable.
	Define the stochastic process $X_n:=\eta(\nu_n)$ for $n\ge 0$.
	From the first two conditions, we have that $X_n\ge K$ for all $n$.
	Note that the second condition specifies that the value of the process is no less than $K$ at the point the loop terminates.
	Then from the third condition and the ``role of independence'' property on \cite[Page 90]{williams1991probability}, we have that
	the condition (C3) holds.
	Thus, by Theorem~\ref{thm:main}, we have $\probm(T\geq n)\leq \alpha^{X_0-K}/\beta^n$.
\end{proof}

\subsection{A Synthesis Algorithm for Exponential Supermartingales}
\label{prevsection}

Reasoning about exponential terms is an inherently challenge for algorithmic analysis. Here we provide a practical synthesis algorithm for exponential supermartingales in Theorem~\ref{thm:sound4} through ranking-supermartingale (RSM) synthesis. Our heuristics is that we treat the function $\eta$ as an RSM-map, which is the core in existing RSM-synthesis algorithms.
Then based on the synthesized RSM-map $\eta$, we resolve $\beta,\alpha$ through polynomial solvers such as MATLAB.

\noindent{\em The Loops We Consider.} We consider that every assignment in the loop is \emph{incremental}, i.e.,~(i)~every assignment has the form $x:=x+e$ where $x$ is a program variable and $e$ is a linear expression consisting of constants and sampling variables.
Although incremental assignments are in limited form, they are expressive enough to implement random walks, gambler's ruin and many other examples in the literature.
Since we utilize the synthesize linear RSMs, we also assume that the loop guard $G$ can be expressed as a polyhedron (i.e., a conjunction of linear inequalities).
Also recall that we only handle samplings that observes discrete probability distributions with bounded range.

Before we illustrate the synthesis algorithm, we first recall the notion of RSM-maps, where we adopt the simplified version for a single probabilistic loop.
Below, we fix an input probabilistic while loop $P$ in the form~(\ref{eq:probprogramform}). Recall that $F$ is the update function.

\begin{definition}[RSM-maps]~\label{def:RSMmaps}
	An \emph{RSM-maps} is a Borel measurable function $\eta:\val{V_\mathrm{p}}\rightarrow \Rset$ such that there exist constants $\epsilon>0, K,K'\le 0$ such that
	\begin{compactitem}
		\item $\mathrm{(B1):}$ $\eta(\nu)\geq 0$ for all $\nu\in G$;
		\item $\mathrm{(B2):}$ $K\le \eta(\nu,\mu)\leq K'$ for all $\nu\in G$ and $\mu\in V_{\mathrm{r}}$;
		\item $\mathrm{(B3):}$ $\sum_{\mu\in\val{V_{\mathrm{r}}}}\sampdpd(\mu) \eta(F(\nu,\mu))\leq \eta(\nu)-\epsilon$  for every $\nu\in G$.
	\end{compactitem}
\end{definition}

Our algorithm fixes the $\epsilon$ to be $1$ in the definition as the function $\eta$ can be scaled by any factor.
Moreover, in our algorithm, the constant $K$ will correspond to the same $K$ in (A1) and (A2).
We do not use $K'$ in our algorithm.

\smallskip
\noindent{\em The Synthesis Algorithm.} The algorithm first synthesizes a linear RSM-map $\eta$ through existing algorithms(which is not the focus of this paper).
Then the algorithm searches the value for $\beta$ through a binary search. Once a value for $\beta$ is chosen, we solve the minimum value for $\alpha$ (so as to obtain the best bound from Theorem~\ref{thm:main}) through the MATLAB polynomial solver. A key point here is that since we only consider incremental assignments and the RSM-map $\eta$ is linear, the condition (A3) reduces to a polynomial inequality that only involves $\alpha$ when the value for $\beta$ is chosen. The algorithm thus proceeds by searching larger and larger $\beta$ through binary search until a given precision is reached, and then finds the corresponding minimal $\alpha$.
Then we apply Theorem~\ref{thm:sound4} to obtain the concentration bound.

\subsection{Prototype Implementation and Experimental Results}\label{result}

\paragraph{Implementation} We implemented the algorithm of Section \ref{prevsection} using Matlab. All results were obtained on a machine with an Intel Core i7-8700 processor (3.2 GHz) and 16 GB of memory, running Microsoft Windows 10.

\paragraph{Benchmarks} We use benchmarks from~\cite{DBLP:journals/toplas/ChatterjeeFNH18,pldi18,ijcai18}. Our benchmarks are as follows:
\begin{itemize}
	\item \problem{AdvRW1D}: This program models a 1-dimensional adversarial random walk and is taken from~\cite{DBLP:journals/toplas/ChatterjeeFNH18}.
	\item \problem{rdwalk1, rdwalk2, rdwalk3}: These programs model skewed 1-dimensional random walks, in which the probabilities of going to the left and to the right at each step are not equal. They are taken from~\cite{pldi18}. In \problem{rdwalk1} there is a $0.75$ probability of moving left and a $0.25$ probability of moving right. In \problem{rdwalk2} the probabilities are $0.875, 0.125$, and in \problem{rdwalk3}, they are $0.9375,0.0625$.
	\item \problem{prspeed}: This example is also a random walk taken from~\cite{pldi18}. The main difference is that the step length is not necessarily $1$.
	\item \problem{mini-Rou, mini-Rou2}: These programs are taken from~\cite{ijcai18}. They are implementations of the mini-roulette casino game as probabilistic programs. \problem{mini-Rou} is the benchmark in~\cite{ijcai18}, while~\problem{mini-Rou2} is a variant where the probabilities of losing gambles are increased in order to reduce the runtime.
	
\end{itemize}

\newcolumntype{H}{>{\setbox0=\hbox\bgroup}c<{\egroup}@{}}

\begin{table}[!h]
	\centering
	\caption{Our Experimental Results over Probabilistic Programs.}
	\label{table:4res}
	\scalebox{0.80}{
		\begin{tabular}{|c|c|c|c|c|c|Hc|c|}
			\hline
			\textbf{Program} & $ \alpha $ & $ \beta $ & $ \kappa $ & \makecell{\textbf{Our bound for}\\$ P(T>\kappa) $} & \makecell{\textbf{\cite{DBLP:journals/toplas/ChatterjeeFNH18}'s}\\\textbf{bound for} $P(T>\kappa)$} &  $ \beta _{\mbox{Hoeffding}}$  & \textbf{RSM} ($\eta$) & \makecell{\textbf{Our} \\\textbf{runtime (s)}} \\
			\hline \hline
			
			\multirow{5}{*}{vad1D } & \multirow{5}{*}{1.1915 } & \multirow{5}{*}{1.1003}
			&  $ 6X_0 $ & $ 2.25 \times 10^{-1} $ & $ 8.16 \times 10^{-1} $ & \multirow{5}{*}{1.027}  & \multirow{5}{*}{$2.8571x $}& \multirow{5}{*}{4.68} \\
			\cline{4-6}
			& &  &  $ 10X_0 $ & $ 3.33 \times 10^{-2} $ & $ 5.12 \times 10^{-1} $ & & &  \\
			\cline{4-6}
			& &  &  $ 14X_0 $ & $ 4.90 \times 10^{-3} $ & $ 3.07 \times 10^{-1} $ & & &  \\
			\cline{4-6}
			& &  &  $ 18X_0 $ & $ 7.28 \times 10^{-4} $ & $ 1.81 \times 10^{-1} $ & & &  \\
			\cline{4-6}
			&  & &   $ 22X_0 $ & $ 1.08 \times 10^{-4} $ & $ 1.06 \times 10^{-1} $ & & &  \\
			\hline
			
			\multirow{5}{*}{rdwalk1 } &\multirow{5}{*}{1.314 } & \multirow{5}{*}{1.1547 }
			&  $ 6X_0 $ & $ 9.07 \times 10^{-2} $ & $ 2.10 \times 10^{-1} $  & \multirow{5}{*}{1.133} & \multirow{5}{*}{$ 2x  $}& \multirow{5}{*}{3.82} \\
			\cline{4-6}
			& &  &  $ 10X_0 $ & $ 5.11 \times 10^{-3} $ & $ 2.06 \times 10^{-2} $ & & &  \\
			\cline{4-6}
			& &  &  $ 14X_0 $ & $ 2.88 \times 10^{-4} $ & $ 1.82 \times 10^{-3} $ & & &  \\
			\cline{4-6}
			&  & &  $ 18X_0 $ & $ 1.62 \times 10^{-6} $ & $ 1.56 \times 10^{-4} $ & & &  \\
			\cline{4-6}
			&  & &  $ 22X_0 $ & $ 9.13 \times 10^{-7} $ & $ 1.31 \times 10^{-5} $ & & &  \\
			\hline
			
			\multirow{5}{*}{rdwalk2 } &		\multirow{5}{*}{2.072 } & \multirow{5}{*}{1.511 }
			&  $ 6X_0 $ & $ 4.16 \times 10^{-4} $ & $ 7.92 \times 10^{-3} $ & \multirow{5}{*}{1.324} & \multirow{5}{*}{$\frac{3}{4}x  $}& \multirow{5}{*}{4.34} \\
			\cline{4-6}
			&	&  &  $ 10X_0 $ & $ 1.07 \times 10^{-7} $ & $ 3.41 \times 10^{-5} $ & & &  \\
			\cline{4-6}
			& &  &  $ 14X_0 $ & $ 2.75 \times 10^{-11} $ & $ 1.32 \times 10^{-7} $ & & &  \\
			\cline{4-6}
			& &  &  $ 18X_0 $ & $ 7.05 \times 10^{-15} $ & $ 4.97 \times 10^{-10} $ & & &  \\
			\cline{4-6}
			&  & &  $ 22X_0 $ & $ 1.81 \times 10^{-18} $ & $ 1.84 \times 10^{-12} $ & & &  \\
			\hline
			
			\multirow{5}{*}{rdwalk3 } &		\multirow{5}{*}{3.265 } & \multirow{5}{*}{2.065}
			&  $ 6X_0 $ & $ 5.07 \times 10^{-7} $ & $ 7.79 \times 10^{-4} $ & \multirow{5}{*}{1.466} & \multirow{5}{*}{$1.142x  $}& \multirow{5}{*}{5.74} \\
			\cline{4-6}
			& &  &  $ 10X_0 $ & $ 2.54 \times 10^{-13} $ & $ 4.39 \times 10^{-7} $ & & &  \\
			\cline{4-6}
			& &  &  $ 14X_0 $ & $ 1.27 \times 10^{-19} $ & $ 2.24 \times 10^{-10} $ & & &  \\
			\cline{4-6}
			&  & &  $ 18X_0 $ & $ 6.35 \times 10^{-26} $ & $ 1.10 \times 10^{-13} $ & & &  \\
			\cline{4-6}
			&  & & $ 22X_0 $ & $ 3.18 \times 10^{-32} $ & $ 5.35 \times 10^{-17} $ & & &  \\
			\hline
			
			\multirow{5}{*}{mini-Rou } &		\multirow{5}{*}{1.00124 } & \multirow{5}{*}{1.00068}
			&  $ 6X_0 $ & $ 1.00 \times 10^{-1} $ & $ 9.96 \times 10^{-1} $ & \multirow{5}{*}{1.000082} & \multirow{5}{*}{$13x-13  $}& \multirow{5}{*}{7.02} \\
			\cline{4-6}
			& & &  $ 10X_0 $ & $ 9.88 \times 10^{-1} $ & $ 9.99 \times 10^{-1} $ & & &  \\
			\cline{4-6}
			& & &  $ 14X_0 $ & $ 9.74 \times 10^{-1} $ & $ 1.00 \times 10^{-1} $ & & &  \\
			\cline{4-6}
			& & &  $ 18X_0 $ & $ 9.96 \times 10^{-1} $ & $ 9.99 \times 10^{-1} $ & & &  \\
			\cline{4-6}
			& & &  $ 22X_0 $ & $ 9.48 \times 10^{-1} $ & $ 9.98 \times 10^{-1} $ & & &  \\
			\hline
			
			\multirow{5}{*}{mini-Rou2 } &		\multirow{5}{*}{1.607 } & \multirow{5}{*}{1.309}
			&  $ 6X_0 $ & $ 3.90 \times 10^{-3} $ & $ 1.44 \times 10^{-1} $ & \multirow{5}{*}{1.077} & \multirow{5}{*}{$0.323x - 0.323  $}& \multirow{5}{*}{6.68} \\
			\cline{4-6}
			& & &  $ 10X_0 $ & $ 1.77 \times 10^{-5} $ & $ 3.26 \times 10^{-2} $ & & &  \\
			\cline{4-6}
			& & &  $ 14X_0 $ & $ 8.11 \times 10^{-8} $ & $ 7.30 \times 10^{-3} $ & & &  \\
			\cline{4-6}
			& &  &  $ 18X_0 $ & $ 3.71 \times 10^{-10} $ & $ 1.60 \times 10^{-3} $ & & &  \\
			\cline{4-6}
			& & &  $ 22X_0 $ & $ 1.69 \times 10^{-12} $ & $ 3.69 \times 10^{-4} $ & & &  \\
			\hline
			
			\multirow{5}{*}{prspeed } &		\multirow{5}{*}{2.0479 } & \multirow{5}{*}{1.3048}
			&  $ 6X_0 $ & $ 4.90 \times 10^{-1} $ & $ 2.52 \times 10^{-2} $ & \multirow{5}{*}{1.078} & \multirow{5}{*}{$1714-1.714x  $}& \multirow{5}{*}{4.58} \\
			\cline{4-6}
			&  & &  $ 10X_0 $ & $ 2.40 \times 10^{-4} $ & $ 6.00 \times 10^{-3} $ & & &  \\
			\cline{4-6}
			& & &  $ 14X_0 $ & $ 1.17 \times 10^{-5} $ & $ 1.37 \times 10^{-3} $ & & &  \\
			\cline{4-6}
			& & &  $ 18X_0 $ & $ 5.71 \times 10^{-8} $ & $ 3.07 \times 10^{-4} $ & & &  \\
			\cline{4-6}
			& & &  $ 22X_0 $ & $ 2.79 \times 10^{-10} $ & $ 6.84 \times 10^{-5} $ & & &  \\
			\hline
	\end{tabular}}
\label{tab:proproblem}
\end{table}

\paragraph{Results and Discussion} Our experimental results are presented in Table~\ref{table:4res}. We compare our concentration bounds with those obtained using Azuma's inequality and Hoeffding's inequality in~\cite{DBLP:journals/toplas/ChatterjeeFNH18}. As shown in Table~\ref{table:4res} our approach successfully obtains tighter bounds in every case. Moreover, note that our approach is very efficient and can obtain these bounds in a few seconds.

\section{Concentration bounds for probabilistic recurrences}\label{mysect:probrec}


\subsection{Problem Setting and Examples}

\paragraph{PRRs} In this section, we consider general Probabilistic Recurrence Relations (PRRs) as defined in~\cite{DBLP:journals/jacm/Karp94}. A PRR is a relation of the following form:
$$
\begin{matrix}\forall n >1,~~\time(n) = a(n) + \sum_{i=1}^\ell \time(h_i(n)) & & &\time(1) = 0,
\end{matrix}
$$
in which $a(n)$ is a non-negative value and each $h_i(n)$ is a non-negative random variable such that we always have $\sum_{i=1}^\ell h_i(n) \leq n-1$. Intuitively, we think of $\time(n)$ as the time it takes for a randomized divide-and-conquer algorithm to solve an instance of size $n$. The algorithm first performs a preprocessing procedure, which leads to $\loc$ smaller instances of random sizes $h_1(n), \ldots, h_\loc(n).$ It then solves the smaller instances recursively and merges the solutions in a postprocessing phase. We assume that the preprocessing and postprocessing on an instance of size $n$ take $a(n)$ units of time overall. Our goal is to obtain upper bounds on the tail probability $\Pr\left[\time(n^*) \geq \kappa\right]$, where $n^*$ is the size of our initial input.


\begin{example}[\problem{QuickSort}~\cite{hoare1961algorithm}] \label{ex:quicksort1}

One of the most well-known sorting algorithms is \problem{QuickSort}. Given an array with $n$ distinct elements, \problem{QuickSort} first chooses an element $p$ of the array uniformly at random. It then compares all the other $n-1$ elements of the array with $p$ and divides them in two parts: (i)~elements that are smaller than $p$, and (ii)~elements that are larger than $p$. Finally, it recursively sorts each part. Let $\time(n)$ be the total number of comparisons made by \problem{QuickSort} in handling an array of size $n$. Assuming that there are $h_1(n)$ elements in part (i) and $h_2(n)$ elements in part (ii) above, we have:
$$
\time(n) = n - 1 + \time(h_1(n)) + \time(h_2(n)).
$$
Moreover, it is easy to see that $h_1(n)+h_2(n)=n-1.$
	
%
\end{example}

\begin{example}[\problem{QuickSelect}~\cite{quickselect}] \label{ex:quickselect1}
	 Consider the problem of finding the $d$-th smallest element in an unordered array of $n$ distinct items. A classical algorithm for solving this problem is \problem{QuickSelect}. Much like \problem{QuickSort}, \problem{QuickSelect} begins by choosing an element $p$ of the array uniformly at random. It then compares all the other $n-1$ elements of the array with $p$ and divides them into two parts: (i)~those that are smaller than $p,$ and (ii)~those that are larger. Suppose that there are $d'$ elements in part (i). If $d'<d-1,$ then the algorithm recursively searches for the $(d-d'-1)$-th smallest element of part (ii). If $d'=d-1,$ the algorithm terminates by returning $p$ as the desired answer. Finally, if $d'>d-1,$ the algorithm recursively finds the $d-$th smallest element in part (i). Note that the classical median selection algorithm is a special case of \problem{QuickSelect}. While more involved linear-time non-randomized algorithms exist for the same problem, \problem{QuickSelect} provides a simple randomized variant.
	 In this case, we have the following PRR:
	$$
	\time(n) = n-1 + \time(h(n)).
	$$
Here, $\time(n)$ is the number of comparisons performed by \problem{QuickSelect} over an input of size $n,$ and $h(n)$ is a random variable that captures the size of the remaining array that has to be searched recursively.

%
\end{example}

\subsection{Modeling and Theoretical Results} \label{sec:rppmodel}

\paragraph{The Markov Chain of a PRR} In order to apply our general approach to PRRs, we first need to embed them in stochastic processes. Suppose we are given a PRR of the following form:
\begin{equation} \label{eq:prr}
\begin{matrix}\forall n >0,~~\time(n) = a(n) + \sum_{i=1}^\ell \time(h_i(n)) & & &\time(1) = 0.
\end{matrix}
\end{equation}
Moreover, assume that we are interested in $\time(n^*)$ for a specific initial value $n^*$. We model this PRR as a Markov chain $(X_m)_{m \geq 0}$ in which each state $X_i$ consists of a non-negative integer $k_i$, and a stack of $k_i$ additional non-negative integers. Formally, for every $i$, we have $X_i = (k_i, \langle n_1^{(i)}, n_2^{(i)}, \ldots, n_{k_i}^{(i)} \rangle).$ Intuitively, each $X_i$ models a state of our probabilistic divide-and-conquer algorithm in which there are $k_i$ more recursive calls waiting to be executed, and the $j$-th call needs to process an instance of size $n_j^{(i)}.$ Following this intuition, the transitions of $(X_m)_{m\geq 0}$ are defined as follows:
\begin{align*}
X_m = \begin{cases}
(1, \langle n^*\rangle) & m = 0\\
(k_{m-1} - 1, \langle n_2^{(m-1)}, \ldots, n_{k_{m-1}}^{(m-1)} \rangle) & m > 0 \land n_1^{(m-1)}=1 \\
(k_{m-1} + \ell  -1, \langle h_1(n_1^{(m-1)}), \ldots ,h_{\ell}(n_1^{(m-1)}), n_2^{(m-1)}, \ldots, n_{k_{m-1}}^{(m-1)} \rangle) & m > 0 \land n_1^{(m-1)}>1 \\
\end{cases}
\end{align*}
Basically, we start with the state $(1, \langle n^*\rangle).$ In other words, in the beginning, we have to process an instance of size $n^*$. At each step in the Markov chain, if $X_{m-1}$ contains a call to process an instance of size $1$, we can easily remove that call from the stack when transitioning to $X_m$, because we assumed that $\time(1) = 0.$ Otherwise, we have to perform a call on an instance of size $n_1^{(m-1)},$ which by definition leads to further recursive calls to instances of sizes $h_1(n_1^{(m-1)}), \ldots, h_\ell(n_1^{(m-1)}).$


\paragraph{Stopping Time and Costs} Let $\tau$ be the stopping time of $(X_m)_{m \geq 0}.$ Formally, $\tau := \min \{i~\vert~k_i = 0\}.$
We further define $(C_m)_{m\geq 0}$ to model the total cost of execution until reaching $X_m$:
\begin{align*}
C_m = \begin{cases}
0 & m = 0\\
C_{m-1} & m > 0 \land n_1^{(m-1)}=1 \\
C_{m-1} + a\left(n_1^{(m-1)}\right) & m > 0 \land n_1^{(m-1)}>1 \\
\end{cases}
\end{align*}
Following the definitions, it is easy to see that $C_{\tau}$ is the total time cost of running the divide-and-conquer algorithm on an instance of size $n^*$.

\paragraph{The Exponential Supermartingale of a PRR} Suppose that we are given a function $f: \mathbb{N} \rightarrow [0, \infty),$ such that $f(n) \geq \expv[\time(n)]$ for all $n \in \mathbb{N}.$ In other words, $f(n)$ is an upper-bound for the expected runtime of our divide-and-conquer algorithm on an instance of size $n$. Finding such upper-bounds is a well-studied problem and can be automated using approaches such as~\cite{DBLP:conf/cav/ChatterjeeFM17}. We define a new stochastic process $(Y_m)_{m \geq 0}$ as follows:
$$
Y_m := {\sum_{j = 1}^{k_m} {f(n_{j}^{(m)})}}.
$$
Moreover, let $\alpha := \alpha(n^*) > 1$ be a constant that depends on the initial input size $n^*$, and define a new stochastic process $(Z_m)_{m \geq 0}$ defined as $Z_m := \alpha^{C_m + Y_m}.$ Note that if $Z_m$ is a supermartingale, then we can obtain a concentration bound for $\time(n^*).$ More formally, we have the following lemma:



\begin{lemma}\label{lemm:recurrence}
	Let $n^*$ be the size of the initial input instance for the recurrence relation in \eqref{eq:prr} and $f: \mathbb{N} \rightarrow [0, \infty)$ an upper-bound on the expected value of $\time.$ If for some $\alpha > 1,$ 
	\begin{equation}\label{eq:recurineq}
	\alpha^{f(n)} \geq \alpha^{a(n)} \cdot \expv\left[\alpha^{\sum_{j=1}^{\ell} f(h_j(n))}\right]
	\end{equation}
	for all $0\le n \le n^*$, then $\Pr\left[\mathcal T(n^*)\geq \kappa\right]\leq \alpha^{f(n^*)-\kappa}$ for all $\kappa \geq f(n^*).$
\end{lemma}

\begin{proof}

	Let $(C_m), (Y_m),$ and $(Z_m)$ be defined as above. Equation~\eqref{eq:recurineq} guarantees that $(Z_m)$ is a supermartingale. Applying Theorem~\ref{thm:main} to $(C_m+Y_m)$ with $\beta :=\alpha,$ we obtain $\Pr\left[\time(n^*) \geq f(n^*) + \kappa'\right] \leq \alpha^{-\kappa'}$ for all $\kappa' \geq 0.$ Substituting $\kappa = f(n^*) + \kappa'$ into the latter inequality leads to the desired concentration bound.
\end{proof}

\subsection{Case Studies} \label{sec:case}

Based on the lemma above, we can now derive concentration bounds for a PRR by synthesizing a suitable $\alpha$. We now demonstrate the process with a few examples. In the next section, we will show how this process can be automated.

\subsubsection{Obtaining Concentration Bounds for \problem{QuickSelect}}
	
Consider the PRR in Example~\ref{ex:quickselect1}. We are interested in the tail probability
$\Pr[\mathcal T(n^*) \geq 12 \cdot n^*]$ and wish to synthesize an upper-bound for this probability. Suppose that the given function $f$ is $f(n):=5\cdot n$. In other words, we know that $\expv\left[\time(n)\right] \leq 5 \cdot n.$ We apply Lemma~\ref{lemm:recurrence} to obtain sufficient conditions on $\alpha$:
\begin{align*}
\forall 1 \leq n\leq n^*,~~~\alpha^{5 \cdot n}&\geq \alpha^{n-1} \cdot \frac{1}{n} \cdot \left(\sum_{i=\lceil n/2 \rceil}^{n-1} \alpha^{5\cdot i} + \sum_{i=\lfloor n/2 \rfloor}^{n-1} \alpha^{5 \cdot i}\right)
\end{align*}
By simply computing the value of the geometric series, we get:
\begin{align*}
\forall 1 \leq n \leq n^*,~~~\alpha^{4\cdot n+1} \geq \frac{1}{n} \cdot \left(\frac{\alpha^{5 \cdot n} - \alpha^{5\cdot\lceil n/2 \rceil} + \alpha^{5 \cdot n} - \alpha^{5\cdot\lfloor n/2 \rfloor}}{\alpha^{5}-1}\right)
\end{align*}
since $\alpha^5 - 1\geq 5 \cdot \ln \alpha$, and $\alpha^{5\cdot\lceil n/2 \rceil} + \alpha^{5\cdot\lfloor n/2 \rfloor} \geq 2\alpha^{2.5\cdot n}$, we strengthen the formula to obtain:
\begin{align*}
\forall 1 \leq n \leq n^*,~~~\alpha^{4\cdot n}&\geq \frac{2}{n} \cdot \left(\frac{\alpha^{5 \cdot n} - \alpha^{2.5\cdot n}}{5 \cdot \ln \alpha}\right)
\end{align*}
Let $c := \alpha^n$. We can rewrite the equation above as:
\begin{align*}
\forall 1 \leq n\leq n^*, ~~~5\cdot c^4 \cdot \ln c&\geq 2 \cdot (c^5-c^{2.5})
\end{align*}
By basic calculus, we can further prove that $5 \cdot c^4 \cdot \ln c\geq 2 \cdot (c^5-c^{2.5})$ holds for $c \in [1,2.74].$  Recall that $c = \alpha^n.$ Since for every $\alpha \geq 1$, $\alpha^n$ increases as $n$ increases, our constraint becomes $1\leq \alpha \land \alpha^{n^*}\leq 2.74$, so one possible solution is $\alpha = (2.74)^{1/n^*}.$ Plugging this value back into Lemma~\ref{lemm:recurrence}, we have $\Pr[\mathcal T(n^*) \geq 12 \cdot n^*]\leq (2.74)^{-7} < 0.0009.$

\begin{remark}[Comparison with~\cite{DBLP:journals/jacm/Karp94}] \label{rem:karp1}
	As shown above, our approach is able to synthesize the concentration bound
	$$\textstyle \Pr[\mathcal T(n^*) \geq 12 \cdot n^*]\leq (2.74)^{-7} < 0.0009$$
	for the PRR corresponding to \problem{QuickSelect}. In contrast,~\cite{DBLP:journals/jacm/Karp94} obtains the following concentration bound:
	$$
	\textstyle \Pr[\mathcal T(n^*) \geq 12 \cdot n^*]\leq \left(\frac{3}{4}\right)^8 \approx 0.1001.
	$$
	Hence, our bound is better by a factor of more than $100$.
\end{remark}

The advantage of our approach is not limited to the specific choice of $12 \cdot n^*$ in the tail probability. See Section~\ref{sec:prr-exp} for concentration bounds for other tail probabilities. We now show how a more general result and a tighter bound can be obtained.

Suppose that we aim to find an upper-bound for the tail probability $\Pr[\time(n^*) \geq r \cdot n^*]$ for an arbitrary constant $r \geq 24$. Let $q > e^2$ be a real number and consider the function $f_q(n) := q \cdot n.$ Using a similar calculation as above, defining $c:=\alpha^n$, we obtain:
$$\forall 1\leq n\leq n^*,~~c^{q/2-1}q \cdot \ln q - 2c^{q/2} + 2\geq 0$$
Since $q>e^2$, the inequality $c^{q/2-1} \cdot q \cdot \ln q - 2c^{q/2} + 2\geq 0$ holds for $c\in [1,q]$, so it suffices to find $\alpha$ such that $\alpha^{n^*}\leq q$. We choose $\alpha=q^{1/n^*}$. Plugging this back into Lemma~\ref{lemm:recurrence}, leads to:
$$\Pr[\mathcal T(n^*)\geq r \cdot n^*]\leq q^{q-r}$$

Specifically, by letting $q = {r}/{\ln r},$ we get
$$\Pr[\mathcal T(n^*)\geq r \cdot n^*]\leq \left(\frac{r}{\ln r}\right)^{\frac{r}{\ln r}-r}.$$ 

\begin{remark}[Comparison with~\cite{DBLP:journals/jacm/Karp94}] \label{rem:karp2}
	If we plug $r = 24$ into the general result above, our general approach is able to synthesize the concentration bound
	$$\textstyle \Pr[\mathcal T(n^*) \geq 24 \cdot n^*]\leq \left(\frac{24}{\ln 24}\right)^{\frac{24}{\ln 24}-24}  < 3.612 \times 10^{-15}$$
	for the PRR corresponding to \problem{QuickSelect}. In contrast,~\cite{DBLP:journals/jacm/Karp94} obtains the following concentration bound:
	$$
	\textstyle \Pr[\mathcal T(n^*) \geq 24 \cdot n^*]\leq \left(\frac{3}{4}\right)^{20} \approx 0.00318.
	$$
	Hence, our bound is better by a factor of more than $8.8 \times 10^{11}$.
\end{remark}

\subsubsection{Obtaining Concentration Bounds for \problem{QuickSort}} \label{manual:quicksort}
Consider the PRR in Example~\ref{ex:quicksort1}. Our goal is to synthesize an upper-bound for the tail probability $\Pr[\mathcal T(n^*) \geq 11\cdot n^* \cdot \ln n^* + 12 \cdot n^*].$ Given $f(n):= 9 \cdot n \cdot \ln n$, we apply Lemma \ref{lemm:recurrence} and obtain the following conditions for $\alpha$:\footnote{We assume $0\ln 0 := 0$.}
$$\forall 1\leq n\leq n^*,~~~\alpha^{9\cdot n \cdot \ln n}\geq \alpha^{n-1} \cdot \frac{1}{n} \cdot \sum_{i = 0}^{n-1}{\alpha^{f(i)+f(n-1-i)}}$$
For $1 \leq n\leq 8$, we can manually verify that it suffices to set $\alpha > 1$. In the sequel, we assume $n \geq 8.$
Note that $\left(i \cdot \ln i + (n-i-1) \cdot \ln(n-i-1)\right)$ is monotonically decreasing on $[1,\lfloor n/2\rfloor]$, and monotonically increasing on $[\lceil n/2\rceil ,n]$. We partition the summation above uniformly into eight parts and use the maximum of each part to overapproximate the sum. This leads to the following upper bound for this summation:
\begin{align*}
\sum_{i = 0}^{n-1}{\alpha^{f(i)}}  & \leq \sum_{j=0}^7 \sum_{i=\lceil j \cdot n / 8 \rceil}^{\lfloor (j+1)\cdot n / 8 \rfloor} \alpha^{f(i)+f(n-1-i)}
\\
&\leq\frac{n}{4} \cdot \left(\alpha^{9\cdot n \cdot \ln n} + \alpha^{9 \cdot (\frac{n}{8} \cdot \ln \frac{n}{8} + \frac{7 \cdot n}{8} \cdot \ln \frac{7 \cdot n}{8})} + \alpha^{9\cdot (\frac{n}{4} \cdot \ln \frac{n}{4} + \frac{3\cdot n}{4}\ln \frac{3 \cdot n}{4})} + \alpha^{9 \cdot (\frac{5 \cdot n}{8}\ln \frac{5 \cdot n}{8} + \frac{3 \cdot n}{8}\ln \frac{3 \cdot n}{8})}\right)
\end{align*}
Plugging in this overapproximation back into the original inequality, we get:
$$\alpha^{9\cdot n \cdot \ln n}\geq \alpha^{n-1} \cdot \frac{1}{4} \cdot \left(\alpha^{9\cdot n \cdot \ln n} + \alpha^{9 \cdot (\frac{n}{8} \cdot \ln \frac{n}{8} + \frac{7 \cdot n}{8} \cdot \ln \frac{7 \cdot n}{8})} + \alpha^{9\cdot (\frac{n}{4} \cdot \ln \frac{n}{4} + \frac{3\cdot n}{4}\ln \frac{3 \cdot n}{4})} + \alpha^{9 \cdot (\frac{5 \cdot n}{8}\ln \frac{5 \cdot n}{8} + \frac{3 \cdot n}{8}\ln \frac{3 \cdot n}{8})}\right)$$
for all $8\leq n\leq n^*$. We define $c:=\alpha^{n}$ to do substitution, and we use the following formula to do strengthening:
\begin{align*}
\alpha^{\beta \cdot \ln \beta + (n-\beta) \cdot \ln (n-\beta)}&\leq \alpha^{\beta \cdot \ln n + (n-\beta) \cdot \ln (n-\beta)}
\\&= \alpha^{n \cdot \ln n} \cdot \alpha^{-(n-\beta) \cdot \ln n + (n-\beta) \cdot \ln (n-\beta)}
\\&= \alpha^{n \cdot \ln n + (n-\beta) \cdot \ln \frac{n-\beta}{n}} = c \cdot c^{\frac{(n-\beta)}{n \cdot \ln n} \cdot \ln \frac{n-\beta}{n}}
\end{align*}
By defining $\beta = \frac n 8, \frac n 4, \frac {3n} 8$ respectively, we obtain:
\begin{align*}
 \forall 8\leq n\leq n^*,~~ &c^{9\ln n}\geq \frac{c}{n}\cdot \frac{n}{4} \cdot \left(c^{9\ln n} + c^{9\ln n + \frac{63}{8} \cdot \ln \frac 78} + c^{9\ln n + \frac{27}{4} \cdot \ln \frac 34} + c^{9\ln n + \frac{45}{8} \cdot \ln \frac 58}\right) \\
 \forall 8\leq n\leq n^*,~~ &4 - \left(c + c^{1- \frac{63}{8} \cdot \ln \frac 87} + c^{1 - \frac{27}{4 } \cdot \ln \frac 43} + c^{1 - \frac{45}{8} \cdot \ln \frac 85}\right) \geq 0 \\
\end{align*}
Now we study the following function:
$$\psi(c) = 4 - \left(c + c^{1- \frac{63}{8} \cdot \ln \frac 87} + c^{1 - \frac{27}{4 } \cdot \ln \frac 43} + c^{1 - \frac{45}{8} \cdot \ln \frac 85}\right).$$
By basic calculus, we can prove that $\psi(c)\geq 0$ holds on $[1,2.3]$. Additionally, since for every $\alpha$, $\alpha^{n}$ increases as $n$ increases, by plugging $\alpha = 2.3^{1/(n^*)}$ into Lemma~\ref{lemm:recurrence}, we obtain:
 $$\Pr\left[\mathcal T(n^*)\geq 11\cdot n^* \cdot \ln n^* + 12 \cdot n^*\right]\leq (2.3)^{-2\cdot \ln n^*-12}.$$

\begin{remark}[Comparison with~\cite{DBLP:journals/jacm/Karp94}] \label{rem:karp3}
	As shown above, our approach is able to synthesize a concentration bound that is of the form
 $$\Pr\left[\time(n^*)\geq 11\cdot n^* \cdot \ln n^* + 12 \cdot n^*\right]\leq (2.3)^{-2\cdot \ln n^*-12}$$
	for the PRR corresponding to \problem{QuickSort}. In contrast~\cite{DBLP:journals/jacm/Karp94} provides the following concentration bound:
	$$
	\Pr\left[\time(n^*) \geq 10 \cdot n^* \cdot \ln n^*\right] \leq e^{-4}.
	$$
	Note that the latter bound is a constant, while our bound improves as $n^*$ grows. More concretely, as $n^*$ grows, our bound becomes exponentially better than the one obtained by~\cite{DBLP:journals/jacm/Karp94}.

%
\end{remark}

Just as in the previous case, the advantage of our approach is not limited to bounding $\Pr[\time(n^*)\geq 11\cdot n^* \cdot \ln n^* + 12 \cdot n^*].$ A similar argument can be applied to obtain similar exponentially-decreasing bounds for $\Pr[\time(n^*)\geq a_1\cdot n^* \cdot \ln n^* + a_2 \cdot n^*]$ with other values of $a_1$ and $a_2.$ Moreover, our bounds improve as $a_2$ increases. This is in contrast to~\cite{DBLP:journals/jacm/Karp94} that can only bound $\Pr\left[ \time(n^*) \geq a \cdot n^* \cdot \ln n^* \right].$ Hence, not only do we beat~\cite{DBLP:journals/jacm/Karp94}'s bounds numerically, but we also provide a more fine-grained set of concentration bounds.

\begin{remark}[Comparison with~\cite{DBLP:journals/corr/Tassarotti17}] \label{rem:joseph}
		Recall that our approach synthesized the following bound:
		 $$\Pr\left[\time(n^*)\geq 11\cdot n^* \cdot \ln n^* + 12 \cdot n^*\right]\leq (2.3)^{-2\cdot \ln n^*-12},$$
for the PRR corresponding to \problem{QuickSort}, the work of~\cite{DBLP:journals/corr/Tassarotti17} improves Karp's cook-book method and provides the following bound: 
		$$
		\Pr\left[ \time(n^*) \geq 11 \cdot n^* \cdot \ln n^* + 12 \cdot n^*\right] \leq (8/7)^{-(2 - \ln (8/7))\cdot\ln n^* -11}
		$$
 
Note that our bound beats theirs in both base and exponent, although both bounds are asympotically equal to $\exp(-\Theta(\ln n^*))$.

\end{remark}

\begin{remark}[Comparison with~\cite{DBLP:journals/jal/McDiarmidH96}] \label{rem:mcd}
		While our approach synthesized the following bound:
		 $$\Pr\left[\time(n^*)\geq 11\cdot n^* \cdot \ln n^* + 12 \cdot n^*\right]\leq (2.3)^{-2\cdot \ln n^*-12},$$
for the PRR corresponding to \problem{QuickSort}, the work of~\cite{DBLP:journals/jal/McDiarmidH96} provides the asympotically optimal bound of the following form:
		$$
		\Pr\left[ \time(n^*) \geq 11 \cdot n^* \cdot \ln n^* + 12 \cdot n^*\right] \leq {e}^{-c \cdot \ln n^* \cdot \ln \ln n^*}
		$$
where $c$ is a constant. Our bound is comparable but slightly worse than the optimal result in \cite{DBLP:journals/jal/McDiarmidH96} by only $\ln \ln n^*$ factor. However, their method only works for quick sort, while our method works on a wide-range of algorithms, such as quick select, quick sort and diameter computation. Furthermore, their method is completetly manual, while our approach could be automated once the monotonic interval of $f(i) + f(n-1-i)$ is obtained, and it is remained as a future work.

\end{remark}

\subsection{Automated Algorithm} \label{sec:autorr}

In this section, we consider the problem of automating our PRR analysis. We provide a sound and polynomial time algorithm. Our algorithm is able to automatically synthesize concentration bounds that beat previous methods over several classical benchmarks.

\paragraph{The Setup} In this section, we consider PRRs of the following form:
\begin{equation} \label{eq:prrauto}
\begin{matrix}\forall n >0,~~\time(n) = a(n) + \time(h(n)) := \mathfrak{t} & & &\time(1) = 0.
\end{matrix}
\end{equation}
in which $\mathfrak{t}$ is an expression generated by the following grammar:
\begin{align*}
\mathfrak t ::= & c\ \mid \ n\ \mid \ln n\ \mid \ n\cdot \ln n\ \mid\  \mathfrak {t} + \mathfrak {t}\ \vert\  c \cdot {\mathfrak t} \\
&\ \mid \ \frac 1n \cdot \sum_{i = 0}^{n-1} \time(j)\ \mid\  \frac{1}{n} \cdot \left(\sum_{i = \lfloor n/2\rfloor}^{n - 1} \time(i) + \sum_{i = \lceil n/2\rceil}^{n-1} \time(i)\right)
\end{align*}
where $c$ is a real constant. We also assume that no sum in $\mathfrak t$ appears with a negative coefficient, and that the coefficient of the most significant non-sum term in $\mathfrak t$ (if it exists) is positive. We focus on the problem of synthesizing upper-bounds for the tail probability $\Pr\left[ \time(n^*) \geq \kappa \right].$ We also assume that the input to our algorithm contains a  function $f: \mathbb{N} \rightarrow [0, \infty)$ that serves as an upper-bound on the expected runtime, i.e.~$f(n) \geq \expv[\time(n)].$ There are well-known approaches for automatically deriving such functions, e.g.~see~\cite{DBLP:conf/cav/ChatterjeeFM17}. To enable algorithmic methods, we further assume that $\kappa$ and $f(n)$ are generated using the following grammar, in which $c$ denotes a real constant:
\begin{equation}
E ::= c\ |\ E + E\  |\  c \cdot E\ |\ n \cdot \ln n\ |\ \ln n\ |\ n \label{gramm:2}
\end{equation}
We also assume that the coefficient of the most significant term in $\kappa$ and $f$ is positive.
The goal of our algorithm is to synthesize an $\alpha > 1$ that satisfies
\begin{equation}\label{eq:recurineq2}
\forall 1 \leq n \leq n^*, ~~\alpha^{f(n)} \geq \alpha^{a(n)} \cdot \expv\left[\alpha^ {f(h(n))}\right].
\end{equation}
Such an $\alpha$ will directly lead to a concentration bound as in Lemma~\ref{lemm:recurrence}. Our algorithm relies on the following simple Lemma:
\begin{lemma} \label{lemma:integration}
	For any monotonically increasing function $f$ defined on the interval $[l,r+1]$, where $l,r\in \mathbb N$, we have:
	$$\sum_{i=l}^r f(x)\leq \int_{l}^{r+1} f(x) \text dx.$$
\end{lemma}\label{lemma:summation}

\paragraph{Overview of the Algorithm} Our algorithm consists of five steps:
\begin{enumerate}
	\item[\emph{Step 1.}] The algorithm symbolically computes Inequality~\eqref{eq:recurineq2}.
	\item[\emph{Step 2.}] The algorithm replaces every summation in~\eqref{eq:recurineq2} with an over-approximation, hence strengthening the requirements. The algorithm has two ways of obtaining such over-approximations: If the expression inside the sum has an elementary antiderivative that can be computed symbolically, the algorithm applies Lemma~\ref{lemma:integration}. Otherwise, it finds a bound based on sampling and relying on monotonicity.
	\item[\emph{Step 3.}] The algorithm introduces a new variable $c := c(\alpha, n)$ and substitutes it into the inequality. It also removes non-significant terms, hence further strengthening the inequality.
	\item[\emph{Step 4.}] The algorithm uses a calculus technique to obtain a value $c^* > 1$ for $c$, such that the inequality holds on $[1, c^*],$ but not on $(c^*, +\infty).$ If no such value is found, the algorithm returns the obvious upper-bound $1$ for the tail probability.
	\item[\emph{Step 5.}] The algorithm plugs $c^*$ back into the definition of $c$ and obtains a value for $\alpha$. Note that this value depends on $n^*$.
\end{enumerate}

\paragraph{Our Synthesis Algorithm} We now present each step of the algorithm in more detail:

\paragraph{Step 1. Computing Conditions on $\alpha$} The algorithm creates a variable $\alpha$ and symbolically computes Inequality~\eqref{eq:recurineq2}.

\begin{example}
	Consider the PRR for \problem{QuickSelect} (Example~\ref{ex:quickselect1}) and assume $f(n):=5 \cdot n.$ The algorithm symbolically computes Inequality~\eqref{eq:recurineq2} and obtains the following:
	\begin{align*}
	\forall 1 \leq n\leq n^*,~~\alpha^{5\cdot n}&\geq \alpha^{n-1} \cdot \frac{1}{n} \cdot \left(\sum_{i=\lceil n/2 \rceil}^{n-1} \alpha^{5\cdot i} + \sum_{i=\lfloor n/2 \rfloor}^{n-1} \alpha^{5 \cdot i}\right).
	\end{align*}
\end{example}

\paragraph{Step 2. Over-approximating the Summations} Note that, by design, the expressions inside our summations are monotonically increasing with respect to the summation index. As such, we can apply Lemma~\ref{lemma:summation} to obtain an upper-bound for each sum. To do so, the algorithm symbolically computes an antiderivative of the expression inside the sums. Also, note that the antiderivative is always concave, given that the initial expression is increasing. The algorithm uses this concavity property to remove floors and ceilings, hence strengthening the inequality.

However, there are cases where no closed-form or elementary antiderivative can be obtained, e.g.~if the expression is $\alpha^{n \cdot \ln n}.$ In such cases, the algorithm partitions the summation uniformly into $B$ parts (as in Section~\ref{manual:quicksort}) and uses the maximum element of a part to over-approximate each of its elements. Furthermore, to tackle with floors and ceilings in such cases, we would strengthen the inequality by replacing $\lceil n/2\rceil$ to $\lceil (n-1)/2 \rceil$.  The algorithm starts with $B = 2,$ and it doubles $B$ to obtain finer over-approximations and repeats the following steps until it succeeds in synthesizing a concentration bound.

\begin{example}
	Continuing with the previous example, we have
	 $$\sum_{i=l}^{r}{\alpha^{5 \cdot i}}\leq \int_{l}^{r+1} \alpha^{5 \cdot x}\text dx = \frac{\alpha^{5 \cdot (r+1)} - \alpha^{5 \cdot l}}{5 \cdot \ln \alpha}$$
	 The algorithm applies this over-approximation to obtain the following strengthened inequality:
	\begin{align*}
	\forall 1 \leq n\leq n^*,~~\alpha^{5\cdot n}&\geq \alpha^{n-1}\cdot  \frac{1}{n}\cdot  \frac{\alpha^{5\cdot n} - \alpha^{5\cdot \lceil n/2\rceil} + \alpha^{5\cdot n} - \alpha^{5\cdot \lfloor n/2\rfloor}}{5 \cdot \ln \alpha}
	\end{align*}
	The algorithm then further strengthens the inequality by removing the floors and ceilings due to the concavity of $\alpha^{5 \cdot x}$:
	\begin{align*}
	\forall 1 \leq n\leq n^*,~~\alpha^{5\cdot n}&\geq \alpha^{n-1}\cdot  \frac{2}{n}\cdot  \frac{\alpha^{5\cdot n} - \alpha^{2.5\cdot n} }{5\ln \alpha}.
	\end{align*}
\end{example}

\paragraph{Step 3. Substitution and Simplification} Let $g(n)$ be the most significant term in $f(n),$ ignoring constant factors. The algorithm defines a new variable $c = c(\alpha, n) := \alpha^{g(n)} > 1,$ and rewrites the inequality based on $c$. It then moves everything to the LHS, writing the inequality in the form of $\mathfrak{e} \geq 0.$ It also eliminates all fractions by multiplying the inequality by their denominators. This is sound, because by construction, the denominators of all fractions are positive at this point. Then, the algorithm inspects all the terms in $\mathfrak{e}.$ If a term is of the form $\mathfrak{e}_1 \cdot c^{\mathfrak{e}_2}$ in which $\mathfrak{e}_2$ contains $n$ as a sub-expression, if $\mathfrak{e}_1 \cdot \mathfrak{e}_2> 0,$ then it can be simplified to $\mathfrak{e}_1$, if$\mathfrak{e}_1 \cdot \mathfrak{e}_2< 0,$ then we check whether $\mathfrak{e}_2\leq 1$ holds, if it holds, it would be simplified into $c$. This preserves soundness and strengthens the inequality. This preserves soundness and strengthens the inequality. Our algorithm eagerly applies as many such simplifications as possible. Finally, the algorithm divides $\mathfrak{e}$ by the greatest common divisor of its terms.



\begin{example}
	
	Continuing with the previous example, we have $f(n) = 5 \cdot n,$ so the most significant term in $f(n)$ is $n$. The algorithm therefore defines $c:=\alpha^n$, and rewrites the inequality as follows: (Note that $\ln c = n \cdot \ln \alpha.$)
	\begin{align*}
	\forall 1 \leq n\leq n^*,~~c^{5}&\geq c^{1 - 1/n} \cdot  \frac{2\cdot c^5 - 2 \cdot c^{2.5}}{5 \cdot \ln c}
	\end{align*}
	It then moves everything to the LHS,
	and eliminates the fractions by multiplying their denominators:
	\begin{align*}
	\forall 1 \leq n\leq n^*,~~5 \cdot c^{5}\cdot \ln c- 2\cdot c^{1 - 1/n}\cdot (c^5 - c^{2.5})\geq 0
	\end{align*}
	Note that $c^{-1/n}<1$ appears on the LHS with negative coefficient, so its removal would strengthen the inequality. The algorithm simplifies the corresponding term, obtaining the following:
	\begin{align*}
	\forall 1 \leq n\leq n^*,~~5 \cdot c^{5}\cdot \ln c-  2\cdot (c^6 - c^{3.5})\geq 0
	\end{align*}
	The algorithm divides the inequality by $c^{3.5},$ obtaining:
	$$
	\forall 1\leq n\leq n^*,~~5 \cdot c^{1.5} \cdot \ln c - 2 \cdot c^{2.5} + 2 \geq 0
	$$
\end{example}

Before moving to the next step, we need two simple lemmas:
\begin{lemma}[Proof in Appendix~\ref{proof:form}]  \label{lemma:form}
	After Steps~1--3 above, our inequality is simplified to $$\forall 1 \leq n \leq n^*,~~ \psi(c) \geq 0$$ in which $\psi(c)$ is a univariate function of the following form:
	\begin{equation}
	\label{eq:form}
	\psi(c) = \sum_{i\in \mathcal I} \mu_i \cdot c^{\nu_i} \ln^{\xi_i} c
	\end{equation}
	where $\mu_i, \nu_i \in \mathbb R$ and $\xi_i\in \{0,1\}$. Also, note that this form is closed under derivation.
\end{lemma}

Given that our inequality is now univariate in $c = \alpha^{g(n)}$, we should look for a value $c^* > 1,$ such that $\psi(c) \geq 0$ on $[1, c^*].$ Intuitively, if we have such a $c^*$, then we can let $\alpha := \left(c^*\right)^{1/g(n)}$ to solve the problem. Moreover, to find the best possible concentration bound, we would like to find the largest possible $c^*$. To simplify the matter, we attempt to obtain a $c^*$ such that $\psi(c) \geq 0$ on $[1, c^*]$ and $\psi(c) < 0$ on $(c^*, +\infty).$ Moreover, we say that $\psi$ is \emph{separable} (into non-negative and negative parts) iff such a $c^*$ exists.
The following lemma provides sufficient conditions for separability:

\begin{lemma}[Proof in Appendix~\ref{proof:Posneg}] \label{lemma:PosNeg}
	Let $\psi(c)$ be a function of the form~\eqref{eq:form}. If at least one of the following conditions holds, then $\psi$ is separable:
	\begin{enumerate}[(i) ]
		\item $\psi$ is strictly decreasing over $[1, +\infty)$ and $\psi(1) \geq 0$.
		\item $\psi'$ is separable and $\psi(1) \geq 0$.
		\item $\psi/c^a$ is separable for some constant $a \in \mathbb R$.
	\end{enumerate}
\end{lemma}

\paragraph{Step 4. Ensuring Separability and Finding $c^*$}
The algorithm attempts to prove separability of $\psi$ using Lemma~\ref{lemma:PosNeg}. Rule (iii) of the Lemma is always used to simplify the expression $\psi$.  The algorithm first evaluates $\psi(1),$ ensuring that it is non-negative. Then, it computes the derivative $\psi'$. If the derivative is negative, then case (i) of Lemma~\ref{lemma:PosNeg} is satisfied and $\psi$ is separable. Otherwise, the algorithm tries to recursively prove the separability of $\psi'$ using the same method, hence ensuring case~(ii) of the Lemma. If both cases fail, the algorithm has failed to prove the separability of $\psi$ and returns the trivial upper-bound $1.$ On the other hand, if $\psi$ is proven to be separable, the algorithm obtains $c^*$ by a simple binary search using the fact that for all $c \geq 1,$ we have $\psi(c) < 0 \Leftrightarrow c > c^*.$

\begin{example}
	Continuing with the previous example, we have
	$$
	\psi(c) = 5 \cdot c^{1.5} \cdot \ln c - 2 \cdot c^{2.5} + 2
	$$
	The algorithm evaluates $\psi(1) = 0,$ which is non-negative. Hence, it computes the following derivative:
	$$
	\psi'(c) = 7.5 \cdot c^{0.5} \cdot \ln c+5 \cdot c^{0.5}-5 \cdot c^{1.5}
	$$
	Note that $\psi'(c)$ is not always negative for $c \geq 1$. Hence, the algorithm tries to recursively prove that $\psi'(c)$ is separable. It first simplifies $\psi'$ to obtain:
	$$
	\psi_1(c) = 7.5 \cdot \ln c + 5 - 5 \cdot c
	$$
	Now it tries to prove the separability of $\psi_1.$ It first evaluates $\psi_1(1) = 0,$ and then computes the derivative:
	$$
	\psi_1'(c) = \frac{7.5}{c} - 5
	$$
	Another level of recursion shows that $\psi_1'(1) \geq 0$ and $\psi_1'$ is strictly decreasing over $[1, +\infty).$ So, it is separable. Hence, it is proven that $\psi$ is separable, too. The algorithm performs a binary search and obtains $c^* \approx 2.74.$
	
\end{example}

\paragraph{Step 5. Applying Lemma~\ref{lemm:recurrence}}
Note that for every $\alpha,$ $c := \alpha^{g(n)}$ increases as $n$ increases. Hence, it suffices to find an $\alpha$ such that $\alpha^{g(1)} > 1$ and $\alpha^{g(n^*)} \leq c^*.$ The algorithm calls an off-the-shelf solver to obtain the largest possible $\alpha$ that satisfies these constraints. It then plugs this $\alpha$ into Lemma~\ref{lemm:recurrence} and reports the following concentration bound:
$\Pr\left[\mathcal T(n^*)\geq \kappa\right]\leq \alpha^{f(n^*)-\kappa}$ for all $\kappa \geq f(n^*).$

\begin{example}
	Continuing with previous examples, we had $c = \alpha^n$ and $c^* = 2.74.$ So, the algorithm solves the constraints $\alpha > 1$ and $\alpha^{n^*} \leq 2.74.$ It is easy to see that $\alpha= \left( 2.74 \right)^{1/n^*}$ is the optimal solution. Hence, the algorithm computes and reports the following concentration bound:
	$$\Pr\left[\mathcal T(n^*)\geq \kappa\right]\leq \left(2.74\right)^{5 - \kappa/n^*}$$
	for all $\kappa \geq 5 \cdot n^*.$ This is equivalent to:
	$$\Pr\left[\mathcal T(n^*)\geq \kappa' \cdot n^*\right]\leq \left(2.74\right)^{5 - \kappa'}$$
	for all $\kappa' \geq 5.$ Note that the bound decreases exponentially as $\kappa'$ grows.
\end{example}

\begin{theorem}[Soundness]
	Given a PRR $\time(n) = a(n) + \time(h(n)),$ an expression $\kappa,$ and an upper-bound function $f$ for the expected runtime of $\time,$ all generated by Grammars~\eqref{eq:prrauto} and~\eqref{gramm:2}, any concentration bound
	$$\forall \kappa \geq f(n^*),~~\Pr\left[\mathcal T(n^*)\geq \kappa\right]\leq \alpha^{f(n^*)-\kappa}$$ generated by the algorithm above is a correct bound.
\end{theorem}

\emph{Proof Sketch.} While the algorithm strengthens the inequality at some points, it never weakens it. Hence, any concentration bound found by our algorithm is valid, and the algorithm is sound.

\begin{theorem}
	Assuming fixed bounds on the number of iterations of the binary search, and the approximation parameter $B$, given a PRR $\time(n) = a(n) + \time(h(n)),$ an initial value $n^* \in \mathbb N$, an expression $\kappa,$ and an upper-bound function $f$ for the expected runtime of $\time,$ all generated by Grammars~\eqref{eq:prrauto} and~\eqref{gramm:2}, the algorithm above terminates in polynomial time with respect to the size of input.
\end{theorem}

\emph{Proof Sketch.} Note that each level of recursion in Step 4, i.e.~simplification and derivation, strictly reduces the number of terms in the expression that is being studied. Hence, this step performs linearly many symbolic operations. Moreover, Step~5 can find the optimal value of $\alpha$ in $O(\vert g \vert)$ operations, where $\vert g \vert$ is the length of the expression $g$ (not its value). It is easy to verify that every other step of the algorithm takes polynomial time.

\subsection{Prototype Implementation and Experimental Results} \label{sec:prr-exp}

\paragraph{Implementation} We implemented the algorithm of Section~\ref{sec:autorr} using Python and Mathematica~\cite{mathematica}. We used the SymPy package~\cite{sympy} for symbolic differentiation and integration. All results were obtained on a machine with an Intel Core i7-8700 processor (3.2 GHz) and 16 GB of memory, running on Microsoft Windows 10.

\paragraph{Benchmarks} We experimented with PRRs corresponding to $4$ classical randomized divide-and-conquer algorithms, namely \problem{QuickSelect}, \problem{RandomSearch}, \problem{L1Diameter}, and \problem{L2Diameter}. \problem{QuickSelect} is already described in Section~\ref{sec:rppmodel}. In \problem{RandomSearch}, the input consists of a sorted array $A$ of $n$ distinct items and a key $k$. The goal is to find the index of $k$ in the array, or report that it does not appear. The algorithm randomly chooses an index $i$ and compares $A[i]$ with $k$. If $A[i]$ is larger, it recursively searches the first $i-1$ locations of the array. If they are equal, it terminates and returns $i$. If $k$ is larger, it recursively searches the last $n-i$ locations. In \problem{L1Diameter} and \problem{L2Diameter}, the input is a set $S$ of $n$ points in the 3-d space. The goal is to find the diameter of $S,$ or equivalently, to find two points in $S$ that are farthest apart from each other. The only difference between the two problems is the norm that is used for computing the distance (i.e.~L1 or L2). Chapter~9 of the classical book~\cite{DBLP:books/cu/MotwaniR95} provides randomized divide-and-conquer algorithms for these problems.

\paragraph{PRRs} The PRRs used in our experiments are shown in the table below.
\begin{center}
	\setlength{\extrarowheight}{.5em}
\begin{tabular}{|c|c|}
\hline
 \textbf{Randomized Algorithm} & \textbf{Probabilistic Recurrence Relation} \\[.5em]
 \hline \hline
 \problem{QuickSelect} &  $\textstyle \time(n) = n - 1 + \frac 1 n\left(\sum_{i=\lceil n/2 \rceil}^{n-1} \time(i) + \sum_{i=\lfloor n/2 \rfloor}^{n-1} \time(i)\right)$  \\[.5em]
 \hline
 \problem{RandomSearch} & { $\textstyle \time(n) = 1 + \frac 1 n\left(\sum_{i=\lceil n/2 \rceil}^{n-1} \time(i) + \sum_{i=\lfloor n/2 \rfloor}^{n-1} \time(i)\right)$ }\\[.5em]
 \hline
 \problem{L1Diameter} & { $\textstyle \time(n) = n + \frac 1 n\left(\sum_{i=0}^{n-1} \time(i)\right)$ } \\[.5em]
 \hline
 \problem{L2Diameter} & { $\textstyle \time(n) = n \ln n + \frac 1 n\left(\sum_{i=0}^{n-1} \time(i)\right)$ }\\[.5em]
 \hline
 \end{tabular}
\setlength{\extrarowheight}{0em}
\end{center}

\paragraph{Results} Our experimental results are shown in Tables~\ref{tab:exp1} and~\ref{tab:exp2}. For each recurrence relation, we consider several different tail probabilities, and manually compute and report the concentration bound obtained by the classical method of~\cite{DBLP:journals/jacm/Karp94}. We then provide results of our algorithm using various different functions $f$. Recall that $f(n)$ is an upper-bound for $\expv[\time(n)].$ In cases where our algorithm was unable to find antiderivatives, we also report the parameter $B,$ i.e.~the number of blocks used in over-approximating the summations.

\paragraph{Discussion} As shown in Tables~\ref{tab:exp1} and~\ref{tab:exp2}, our algorithm is very efficient and can handle the benchmarks in a few seconds. Moreover, it obtains concentration bounds that are consistently tighter than those of~\cite{DBLP:journals/jacm/Karp94}, often by one or more orders of magnitude (see the ratio columns). It is also noteworthy that, for the \problem{RandomSearch} benchmark, our algorithm obtains concentration bounds that decrease as $n^*$ goes up, whereas~\cite{DBLP:journals/jacm/Karp94} only provides constant bounds. In this case, the ratio becomes arbitrarily large as $n^*$ grows.

\begin{table}
\caption{Experimental results over the PRRs corresponding to \problem{QuickSelect} and \problem{RandomSearch}}
\setlength{\extrarowheight}{.2em}
\resizebox{\linewidth}{!}{
\begin{tabular}{|c|c||c||c|c|c|c|c|}

\hline
\textbf{Recurrence} & \textbf{Tail Probability} & \textbf{Karp's bound} & $\mathbf{f(n)}$ & \textbf{Our bound} & \makecell{\textbf{Our}\\\textbf{runtime (s)}} & \makecell{\textbf{Ratio of} \\ \textbf{the bounds $\approx$}} & $\mathbf{B}$ \\
\hline \hline
\multirow{15}{*}{\problem{QuickSelect}} &
\multirow{3}{*}{$\Pr[\mathcal T(n^*)\geq 17n^*]$} & \multirow{3}{*}{${\left(\frac34\right)}^{13}\approx 0.024$} & $6n$ & $4.36\cdot 10^{-8}$ & $3.27$ & $5\cdot 10^5$ & \multirow{15}{*}{/}
\\\cline{4-7} & & & $7n$ & $6.11\cdot 10^{-9}$ & $3.24$ & $3.63\cdot 10^6$ &
\\\cline{4-7} & & & $10n$ & $1.86\cdot 10^{-8}$ & $3.34$ & $1.29\cdot 10^6$ &
\\\cline{2-7} & \multirow{3}{*}{$\Pr[\mathcal T(n^*)\geq 15n^*]$} & \multirow{3}{*}{${\left(\frac34\right)}^{11}\approx 0.043$} & $9n$ & $6.89\cdot 10^{-7}$ & $2.32$ & $6.24\cdot 10^4$ &
\\\cline{4-7} & & & $12.5n$ & $0.002$ & $2.18$ & $21.5$ &
\\\cline{4-7} & & & $13n$ & $0.005$ & $2.09$ & $8.6$ &
\\\cline{2-7} & \multirow{3}{*}{$\Pr[\mathcal T(n^*)\geq 11n^*]$} & \multirow{3}{*}{${\left(\frac34\right)}^{7}\approx 0.021$} & $5n$ & $0.0024$ & $2.25$ & $8.75$ &
\\\cline{4-7} & & & $6n$ & $0.0021$ & $2.10$ & $10$ &
\\\cline{4-7} & & & $8n$ & $0.0016$ & $2.15$ & $13.125$ &
\\\cline{2-7} & \multirow{3}{*}{$\Pr[\mathcal T(n^*)\geq 8n^*]$} & \multirow{3}{*}{${\left(\frac34\right)}^{4}\approx 0.317$} & $5.5n$ & $0.039$ & $2.42$ & $8.12$ &
\\\cline{4-7} & & & $6n$ & $0.046$ & $1.82$ & $6.89$ &
\\\cline{4-7} & & & $7n$ & $0.151$ & $2.28$ & $2.10$ &
\\\cline{2-7} & \multirow{3}{*}{$\Pr[\mathcal T(n^*)\geq 6n^*]$} & \multirow{3}{*}{${\left(\frac34\right)}^{2}= 0.5625$} & $4.5n$ & $0.406$ & $2.78$ & $1.39$ &
\\\cline{4-7} & & & $5n$ & $0.365$ & $2.67$ & $1.54$ &
\\\cline{4-7} & & & $5.2n$ & $0.402$ & $3.45$ & $1.39$ &
\\\hline
\multirow{15}{*}{\problem{RandomSearch}} &
\multirow{3}{*}{$\Pr[\mathcal T(n^*)\geq 11\ln n^*]$} & \multirow{3}{*}{${\left(\frac34\right)}^{11-\frac{1}{\ln \frac{4}{3}}}\approx 0.12$} & $5\ln n$ & $(n^*)^{-8.24}$ & $3.12$ & \multirow{15}{*}{\makecell{$+\infty$\\as $n^* \rightarrow \infty$}} & \multirow{15}{*}{/}
\\\cline{4-6} & & & $7\ln n$ & $(n^*)^{-8.11}$ & $3.07$ & &
\\\cline{4-6} & & & $9\ln n$ & $(n^*)^{-6.75}$ & $3.20$ & &
\\\cline{2-6} & \multirow{3}{*}{$\Pr[\mathcal T(n^*)\geq 10\ln n^*]$} & \multirow{3}{*}{${\left(\frac34\right)}^{10-\frac{1}{\ln \frac{4}{3}}}\approx 0.154$} & $7\ln n$ & $(n^*)^{-6.08}$ & $2.28$ & &
\\\cline{4-6} & & & $8.5\ln n$ & $(n^*)^{-3.52}$ & $2.90$ & &
\\\cline{4-6} & & & $9.5\ln n$ & $(n^*)^{-1.26}$ & $2.37$ & &
\\\cline{2-6} & \multirow{3}{*}{$\Pr[\mathcal T(n^*)\geq 8\ln n^*]$} & \multirow{3}{*}{${\left(\frac34\right)}^{8-\frac{1}{\ln \frac{4}{3}}}\approx 0.273$} &
$5.5\ln n$ & $(n^*)^{-3.94}$ & $2.50$ & &
\\\cline{4-6} & & & $6\ln n$ & $(n^*)^{-3.49}$ & $2.37$ & &
\\\cline{4-6} & & & $6.5\ln n$ & $(n^*)^{-2.84}$ & $2.26$ & &
\\\cline{2-6} & \multirow{3}{*}{$\Pr[\mathcal T(n^*)\geq 7\ln n^*]$} & \multirow{3}{*}{${\left(\frac34\right)}^{7-\frac{1}{\ln \frac{4}{3}}}\approx 0.363$} &
$4.5\ln n$ & $(n^*)^{-2.80}$ & $2.29$ & &
\\\cline{4-6} & & & $5.5\ln n$ & $(n^*)^{-2.36}$ & $2.47$ & &
\\\cline{4-6} & & & $6\ln n$ & $(n^*)^{-1.74}$ & $2.40$ & &
\\\cline{2-6} & \multirow{3}{*}{$\Pr[\mathcal T(n^*)\geq 5\ln n^*]$} & \multirow{3}{*}{${\left(\frac34\right)}^{5-\frac{1}{\ln \frac{4}{3}}}\approx 0.645$} & $3.7\ln n$ & $(n^*)^{-0.68}$ & $2.37$ & &
\\\cline{4-6} & & & $4\ln n$ & $(n^*)^{-0.78}$ & $2.61$ & &
\\\cline{4-6} & & & $4.5\ln n$ & $(n^*)^{-0.56}$ & $2.29$ & &
\\\hline
\end{tabular}
}
\setlength{\extrarowheight}{0em}

\label{fig:expresult51}
\label{tab:exp1}

\end{table}

\begin{table}
\caption{Experimental results over the PRRs corresponding to \problem{L1Diameter} and \problem{L2Diameter}}
\setlength{\extrarowheight}{.2em}
\resizebox{\linewidth}{!}{
\begin{tabular}{|c|c||c||c|c|c|c|c|}
\hline
\textbf{Recurrence} & \textbf{Tail Probability} & \textbf{Karp's bound} & $\mathbf{f(n)}$ & \textbf{Our bound} & \makecell{\textbf{Our}\\\textbf{runtime (s)}} & \makecell{\textbf{Ratio of} \\ \textbf{the bounds $\approx$}} &  $\mathbf{B}$ \\
\hline \hline
\multirow{15}{*}{\problem{L1Diameter}} &
\multirow{3}{*}{$\Pr[\mathcal T(n^*)\geq 13n^*]$} & \multirow{3}{*}{${\left(\frac12\right)}^{11}\approx 4.89\cdot 10^{-4}$} & $4.3n$ & $2.33\cdot 10^{-9}$ & $3.76$ & $2.09\cdot 10^{5}$ & \multirow{15}{*}{/}
\\\cline{4-7} & & &  $5n$ & $1.47\cdot 10^{-9}$ & $2.97$ & $3.32\cdot 10^5$ &
\\\cline{4-7} & & & $5.2n$ & $1.48 \cdot 10^{-9}$  & $3.34$ & $3.34\cdot 10^5$ &
\\\cline{2-7} & \multirow{3}{*}{$\Pr[\mathcal T(n^*)\geq 11n^*]$} & \multirow{3}{*}{${\left(\frac12\right)}^{9}= 0.002$} & $2.5n$ & $1.891\cdot 10^{-4}$ & $2.09$ &$10.58$ &
\\\cline{4-7} & & & $6n$ & $1.317\cdot 10^{-6}$ & $1.80$ & $1518.61$ &
\\\cline{4-7} & & & $7n$ & $1.976\cdot 10^{-5}$ & $1.82$ & $101.21$ &
\\\cline{2-7} & \multirow{3}{*}{$\Pr[\mathcal T(n^*)\geq 9n^*]$} & \multirow{3}{*}{${\left(\frac12\right)}^{7}= 0.008$} & $2.5n$ & $0.002$ & $2.33$ & $4$ &
\\\cline{4-7} & & & $5.5n$ & $7.957\cdot 10^{-5}$ & $2.27$ & $25.14$ &
\\\cline{4-7} & & & $6n$ & $2.963\cdot 10^{-4}$ & $2.00$ & $6.75$ &
\\\cline{2-7} & \multirow{3}{*}{$\Pr[\mathcal T(n^*)\geq 7n^*]$} & \multirow{3}{*}{${\left(\frac12\right)}^{5}= 0.032$} & $3.5n$ & $0.002$ & $2.07$ & $16$ &
\\\cline{4-7} & & & $5n$ & $0.007$ & $2.07$ & $4.571$ &
\\\cline{4-7} & & & $5.5n$ & $0.018$ & $2.50$ & $1.778$ &
\\\cline{2-7} & \multirow{3}{*}{$\Pr[\mathcal T(n^*)\geq 5n^*]$} & \multirow{3}{*}{${\left(\frac12\right)}^{3}= 0.125$} & $2.5n$ & $0.081$ & $2.68$ & $1.54$ &
\\\cline{4-7} & & & $3n$ & $0.046$ & $2.78$ & $2.717$ &
\\\cline{4-7} & & & $4n$ & $0.117$ & $2.56$ & $1.068$ &
\\\hline
\multirow{10}{*}{\problem{L2Diameter}} &
\multirow{2}{*}{$\Pr[\mathcal T(n^*)\geq 20n^*\ln n^*]$} & \multirow{2}{*}{${\left(\frac12\right)}^{18}\approx 3.81\cdot 10^{-6}$} & $3.5n\ln n$ & $2.07\cdot 10^{-6}$ & $4.90$ & $1.84$ & $2$
\\\cline{4-8} & & & $5n\ln n$ & $5.51\cdot 10^{-10}$ & $15.42$ & $6914.7$ & $4$
\\\cline{2-8} & \multirow{2}{*}{$\Pr[\mathcal T(n^*)\geq 15n^*\ln n^*]$} & \multirow{2}{*}{${\left(\frac12\right)}^{13}\approx 1.23\cdot 10^{-4}$} &
$5n\ln n$ & $6.715\cdot 10^{-7}$ & $12.41$ & $567.38$ & $4$
\\\cline{4-8} & & & $7n\ln n$ & $1.41\cdot 10^{-5}$ & $14.21$ & $27.02$ & $4$
\\\cline{2-8} & \multirow{2}{*}{$\Pr[\mathcal T(n^*)\geq 13.5n^*\ln n^*]$} & \multirow{2}{*}{${\left(\frac12\right)}^{11.5}\approx 3.45\cdot 10^{-4}$} &
$2.5n\ln n$ & $7.22\cdot 10^{-5}$ & $7.51$ & $5.27$ & $2$
\\\cline{4-8} & & & $5n\ln n$ & $5.67\cdot 10^{-6}$ & $14.75$ & $67.19$ & $4$
\\\cline{2-8} & \multirow{2}{*}{$\Pr[\mathcal T(n^*)\geq 9n^*\ln n^*]$} & \multirow{2}{*}{${\left(\frac12\right)}^{7}=0.008$} &
$2.5n\ln n$ & $0.004$ & $5.64$ & $2$ & $2$
\\\cline{4-8} & & & $4.5n\ln n$ & $0.001$ & $12.44$ & $8$ & $4$
\\\cline{2-8} & \multirow{2}{*}{$\Pr[\mathcal T(n^*)\geq 8n^*\ln n^*]$} & \multirow{2}{*}{${\left(\frac12\right)}^{6}\approx 0.016$} &
$2.5n\ln n$ & $0.009$ & $6.14$ & $1.79$ & $2$
\\\cline{4-8} & & & $4.5n\ln n$ & $0.007$ & $15.13$ & $2.29$ & $4$
\\ \hline
\end{tabular}
}
\setlength{\extrarowheight}{0em}
\label{fig:expresult52}
\label{tab:exp2}

\end{table}

\section{Related Works}

\paragraph{Previous results on concentration bounds for probabilistic programs}
Concentration bounds for probabilistic programs were first considered by~\cite{DBLP:conf/sas/Monniaux01}
where a basic approach for obtaining exponentially-decreasing concentration bounds through abstract interpretation and
truncation of the sampling intervals is presented.
The work of~\cite{SriramCAV} used Azuma's inequality to derive exponentially-decreasing concentration results
for values of program variables in a probabilistic program.
For termination time of probabilistic programs exponentially-decreasing concentration bounds for special
classes of probabilistic programs using Azuma and Hoeffding inequalities were established in~\cite{DBLP:journals/toplas/ChatterjeeFNH18}.
Recently, for several cases where the Azuma and Hoeffding inequalities are not applicable,
a reciprocal concentration bound using Markov's inequality was presented in~\cite{DBLP:journals/corr/ChatterjeeF17},
which was then extended to higher moments in~\cite{DBLP:conf/tacas/KuraUH19,DBLP:journals/corr/abs-2001-10150}.
For a detailed survey of the current methods of concentration bounds for probabilistic programs see~\cite{handbookchapter}.

\paragraph{Previous results on concentration bounds for probabilistic recurrences}
Concentration-bound analyses for probabilistic recurrence relations were first considered in the classical work of~\cite{DBLP:journals/jacm/Karp94},
where cookbook methods (similar to the master theorem for  worst-case analysis of recurrences) were obtained for a large class of probabilistic recurrence relations.
A variant of the results of Karp that weakened several conditions but obtained comparable bounds was presented in~\cite{DBLP:journals/tcs/ChaudhuriD97}.
The optimal concentration bound for the \problem{QuickSort} algorithm was presented in~\cite{DBLP:journals/jal/McDiarmidH96}.
Recently, concentration bounds for probabilistic recurrence relations were mechanized in a theorem prover~\cite{DBLP:conf/itp/Tassarotti018},
and extended to parallel settings~\cite{DBLP:journals/corr/Tassarotti17}.

\paragraph{Comparison with previous approaches on probabilistic programs}
Compared with previous results on probabilistic programs, our approach considers synthesis of exponential supermartingales.
As compared to~\cite{DBLP:conf/sas/Monniaux01}, our approach is based on the well-studied theory of martingales for stochastic
processes.
In comparison to previous martingale-based approaches, we either achieve asymptotically better bounds
(e.g.~we achieve exponentially-decreasing bounds as compared to polynomially-decreasing bounds of~\cite{DBLP:journals/corr/ChatterjeeF17,DBLP:conf/tacas/KuraUH19,DBLP:journals/corr/abs-2001-10150})
or substantially improve the bounds (e.g.~in comparison with~\cite{DBLP:journals/toplas/ChatterjeeFNH18} (see our experimental results in Section~\ref{result}).
Moreover, all previous results, such as~\cite{SriramCAV,DBLP:journals/toplas/ChatterjeeFNH18}, that achieve exponentially-decreasing bounds require bounded-difference, i.e.~the stepwise difference in a supermartingale needs to be globally bounded to apply Azuma or Hoeffding inequalities. In contrast, our results can apply to stochastic processes that are not necessarily difference-bounded. 

\paragraph{Comparison with previous approaches on probabilistic recurrences}
Compared with previous results on probabilistic recurrence relations, our result is based on
the idea of exponential supermartingales and related automation techniques.
It can derive much better concentration bounds over the classical approach of~\cite{DBLP:journals/jacm/Karp94} (See Remarks~\ref{rem:karp1},\ref{rem:karp2}, and \ref{rem:karp3} in Section~\ref{sec:case}), and beats the \yican{manual} approach of~\cite{DBLP:journals/corr/Tassarotti17} in constants(See Remark \ref{rem:joseph}).
Moreover, our approach also derives a tail bound for \problem{QuickSort} that is comparable to the optimal bound proposed in~\cite{DBLP:journals/jal/McDiarmidH96}
(see Remark~\ref{rem:mcd} in Section~\ref{sec:case}).
In addition, the result of~\cite{DBLP:journals/jacm/Karp94} requires the key condition $\sum_{i}\mathbb{E}(T(h_i(n)))\le \mathbb{E}(T(n))$
to handle recurrences with multiple procedure calls.
Whether this condition can be relaxed is raised as an important open problem in~\cite{DBLP:journals/jacm/Karp94,DBLP:books/daglib/0025902}.
To address this issue, the approach of~\cite{DBLP:journals/corr/Tassarotti17} proposed a method that does not require this restriction. Our method also does not need this restriction, hence could be viewed as a new way to resolve this problem.

\paragraph{Key conceptual difference}
Finally, our general approach of exponential supermartingales has a key difference with respect to previous approaches.
The main conceptual difference is that our approach examines the detailed probability distribution (as moment generating function).
In comparison, previous approaches (e.g.~those based Azuma and Hoeffding inequalities, or results of~\cite{DBLP:journals/jacm/Karp94})
only consider the expectation, the variance, or the range of related random variables.
This is the key conceptual reason that our approach can derive tighter bounds.

\section{Conclusion and Future Work}

In this work, we presented a new approach to derive tighter concentration bounds for probabilistic programs and probabilistic recurrence relations.
We showed that our new approach can derive tighter concentration bounds than the classical methods of applying 
Azuma's inequality for probabilistic programs and the classical approaches for probabilistic recurrences,
even for basic randomized algorithms such as \problem{QuickSelect} and randomized diameter computation. \yican{On} \problem{QuickSort}, \yican{we beat the approach of ~\cite{DBLP:journals/jacm/Karp94} and~\cite{DBLP:journals/corr/Tassarotti17}, and derive a bound comparable to the optimal bound in ~\cite{DBLP:journals/jal/McDiarmidH96}}. 

There are several interesting directions for future work.
First, while we consider classical probability distributions such as uniform and Bernoulli for 
algorithmic aspects, extending the algorithms to handle various other distributions is an interesting 
direction.
Second, whether our technique can be used for the relaxation of the key condition $\sum_{i}\mathbb{E}(T(h_i(n)))\le \mathbb{E}(T(n))$ in the approach of~\cite{DBLP:journals/jacm/Karp94} 
is another interesting problem.
Finally, whether our approach can be directly applied to analyze randomized algorithms, rather than their corresponding probabilistic recurrences, is another direction of future work.

\section{Acknowledgements}

We sincerely thank \yican{Prof.} Joseph Tassarotti for pointing out \yican{a miscalculation} in Remark \ref{rem:mcd}.


\appendix

\section{Detailed Experimental Results for Section~\ref{result} }

In this section, there are some figures of related example in section ~\ref{result} in Figure \ref{figure:ddd}. The numbers on the x-axis are the terminating time  $n $ varying from $6X_0$ to $22X_0$.
\begin{figure}[!h] 
	\begin{minipage}[t]{0.4\linewidth}
		\centering 
		\includegraphics[width=1\textwidth]{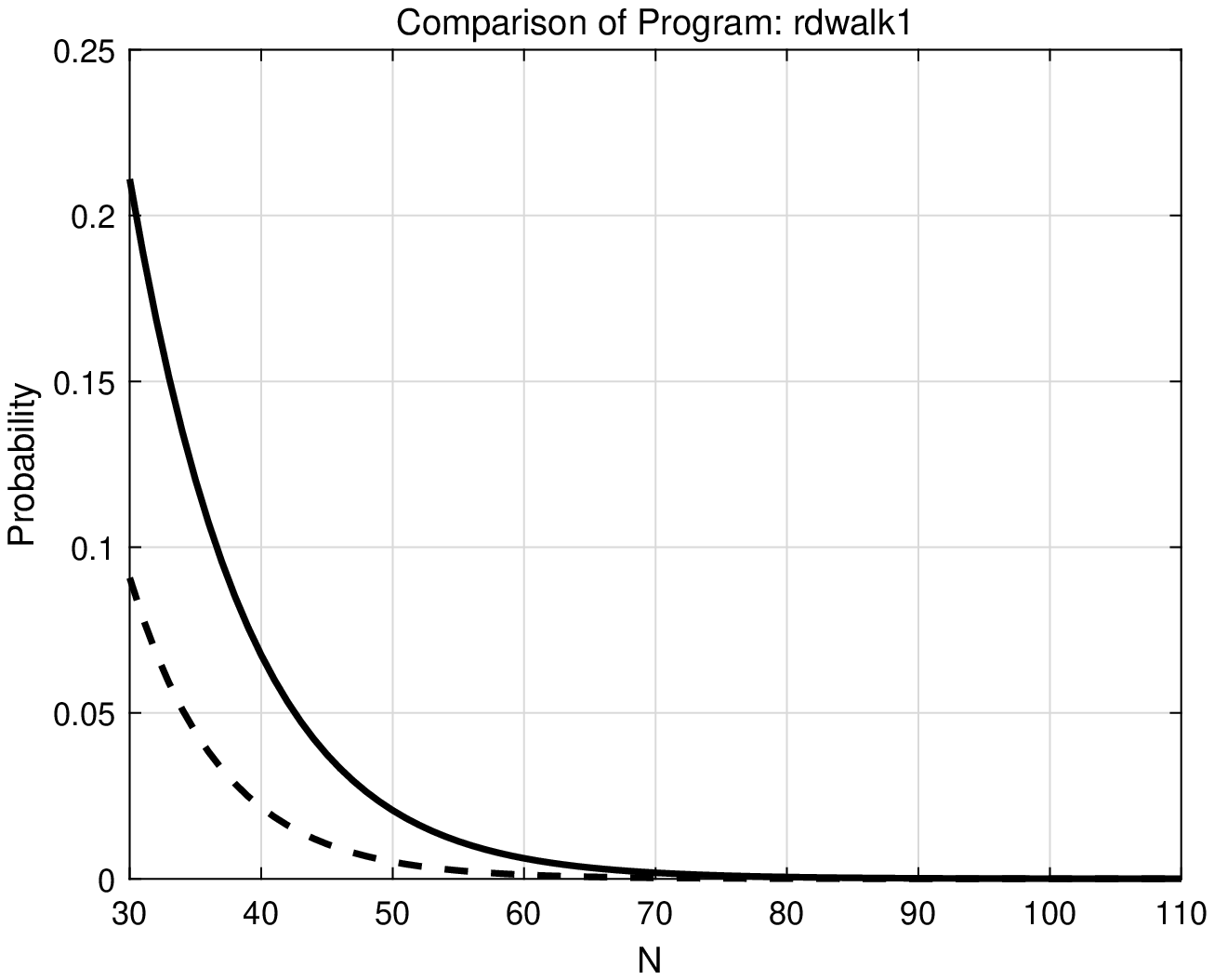} 
		\caption*{(a) This is the result of example rdwalk1.}
		\label{fig:side:a} 
	\end{minipage} 
	\begin{minipage}[t]{0.4\linewidth}
		\centering 
		\includegraphics[width=1\textwidth]{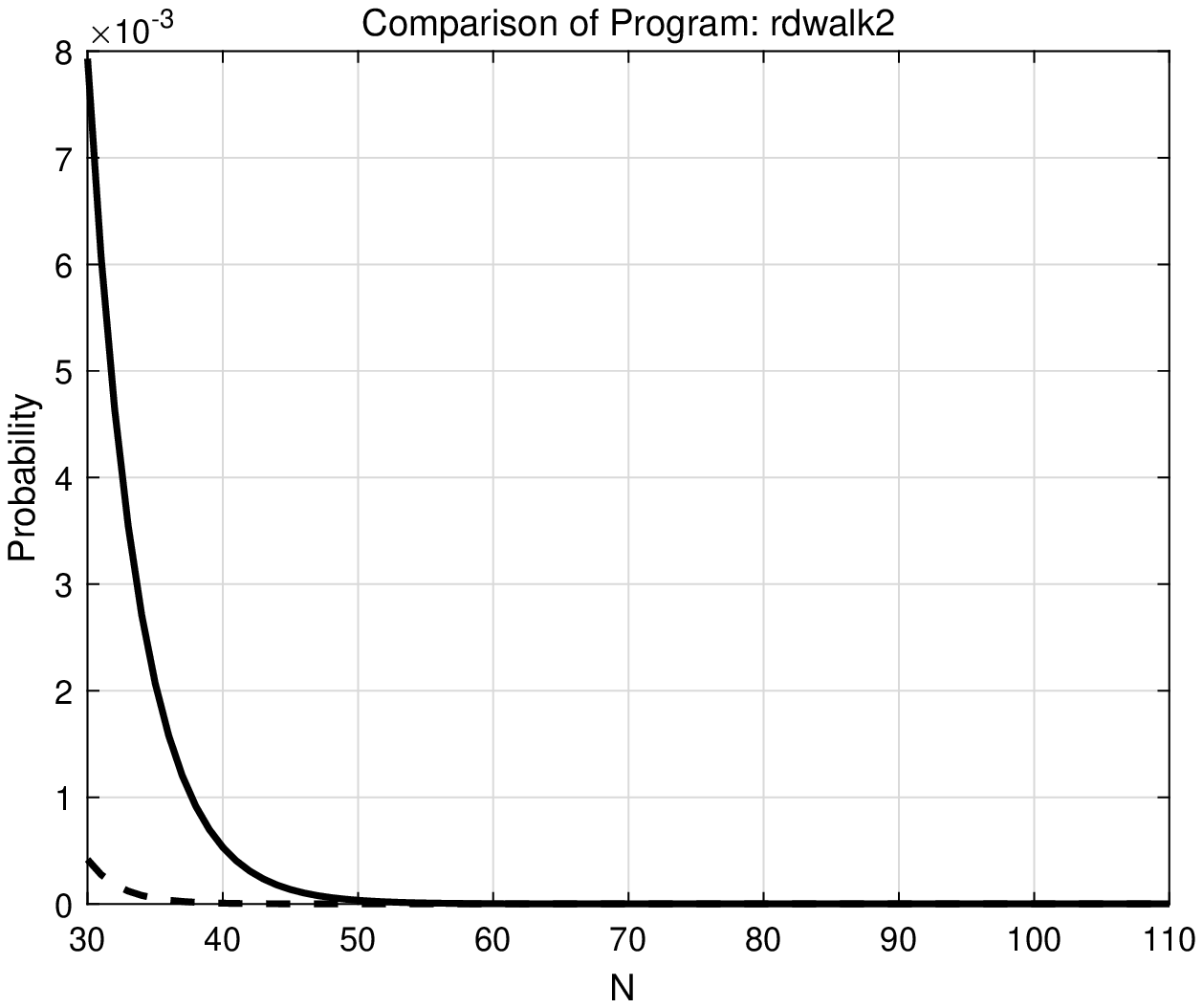} 
		\caption*{(b)This is the result of example rdwalk2.} 
		\label{fig:side:b} 
	\end{minipage}
		\begin{minipage}[t]{0.4\linewidth}
		\centering 
		\includegraphics[width=1\textwidth]{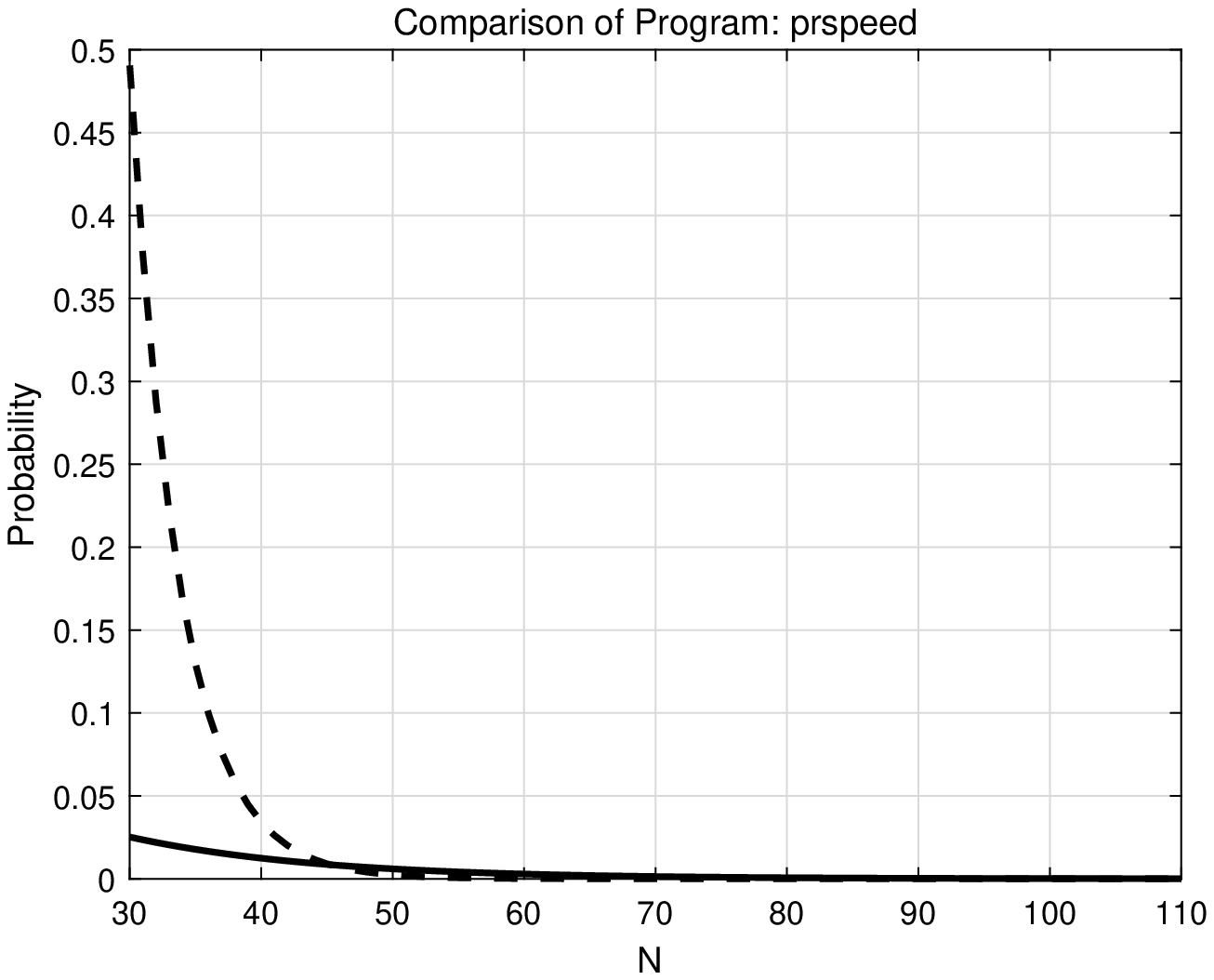} 
		\caption*{(c) This is the result of example prspeed.}
		\label{fig:side:a} 
	\end{minipage} 
	\begin{minipage}[t]{0.4\linewidth}
		\centering 
		\includegraphics[width=1\textwidth]{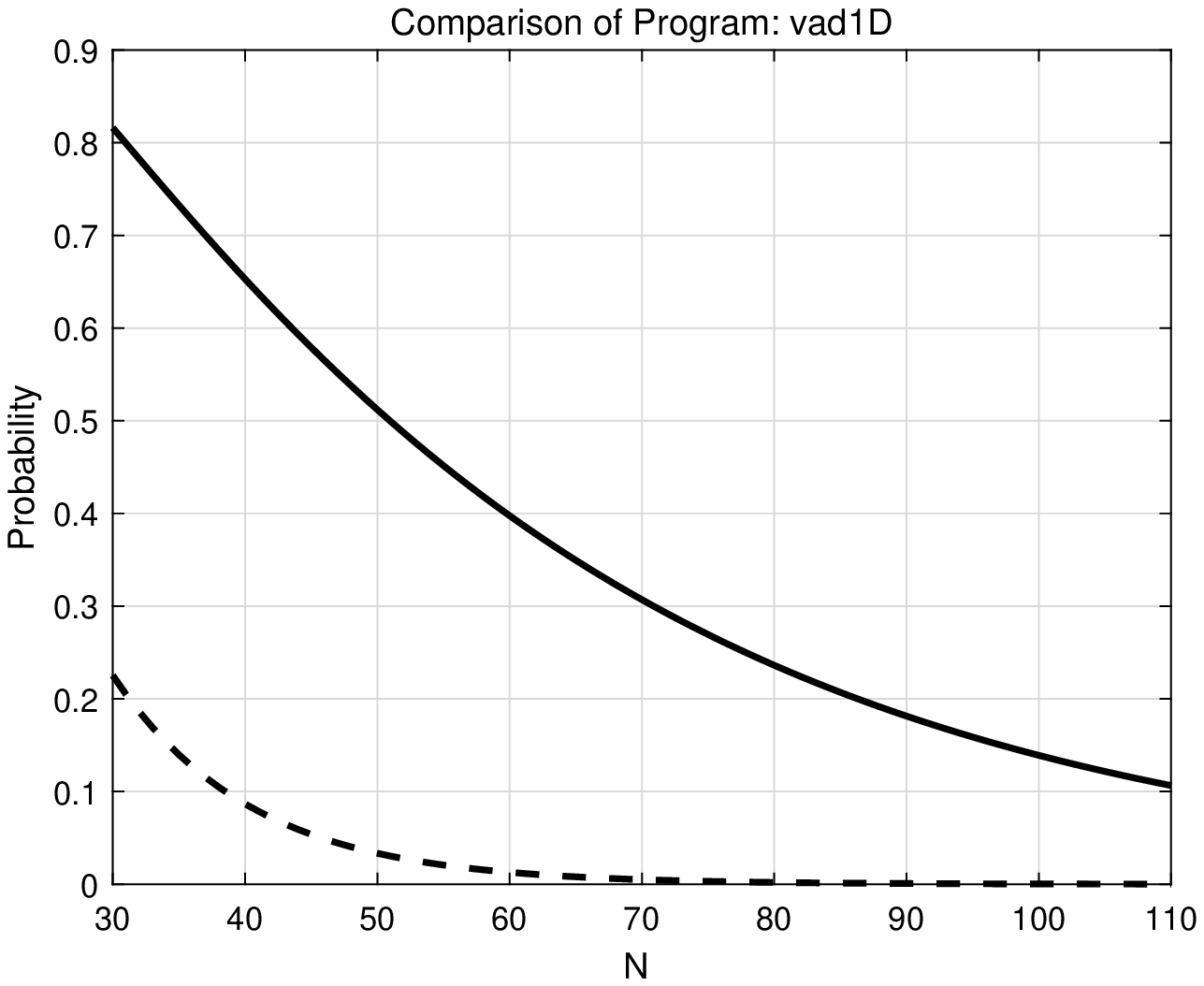} 
		\caption*{(d)This is the result of example vad1D.} 
		\label{fig:side:b} 
	\end{minipage}
		\begin{minipage}[t]{0.4\linewidth}
		\centering 
		\includegraphics[width=1\textwidth]{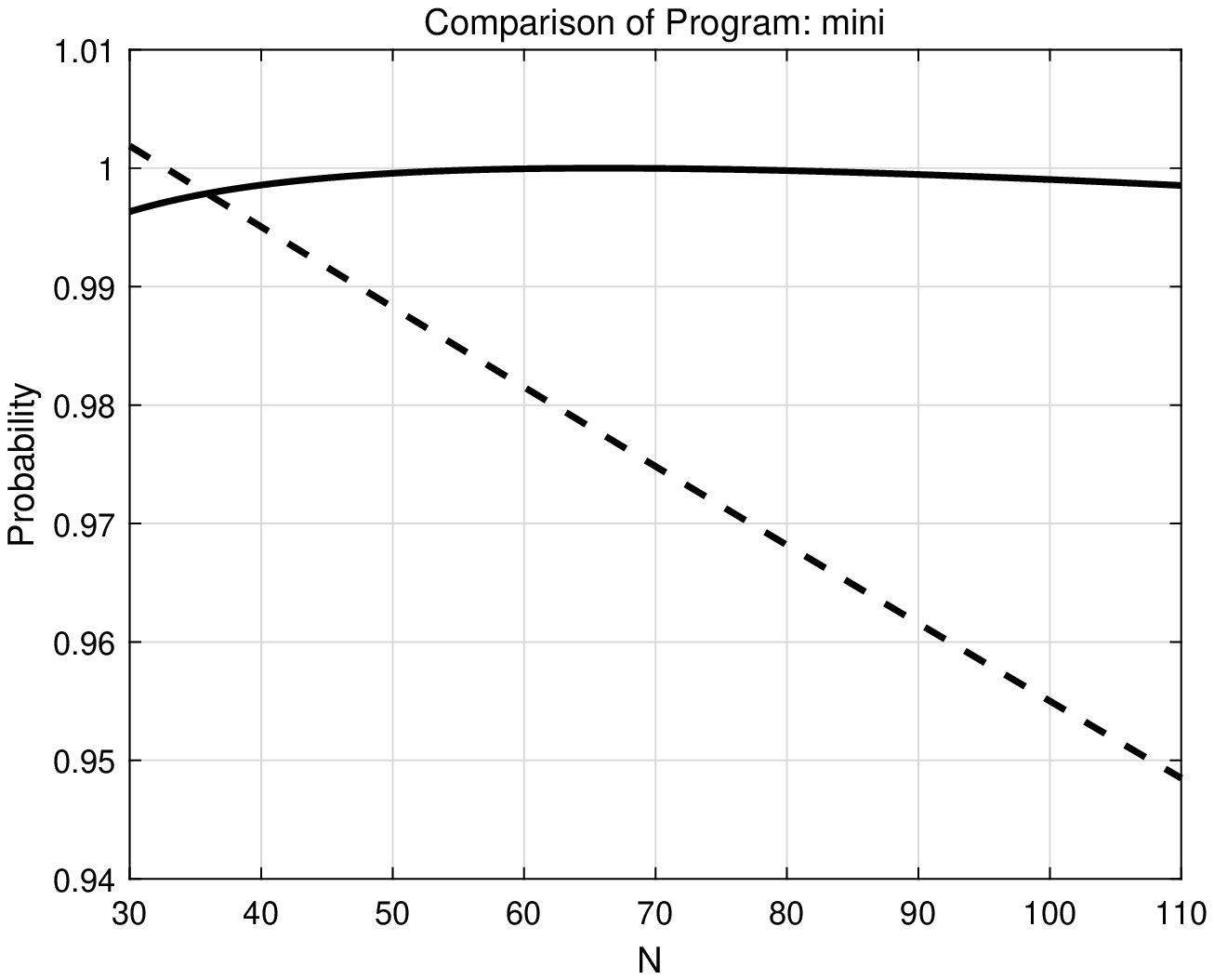} 
		\caption*{(e) This is the result of example mini-roulette.}
		\label{fig:side:a} 
	\end{minipage} 
	\begin{minipage}[t]{0.4\linewidth}
		\centering 
		\includegraphics[width=1\textwidth]{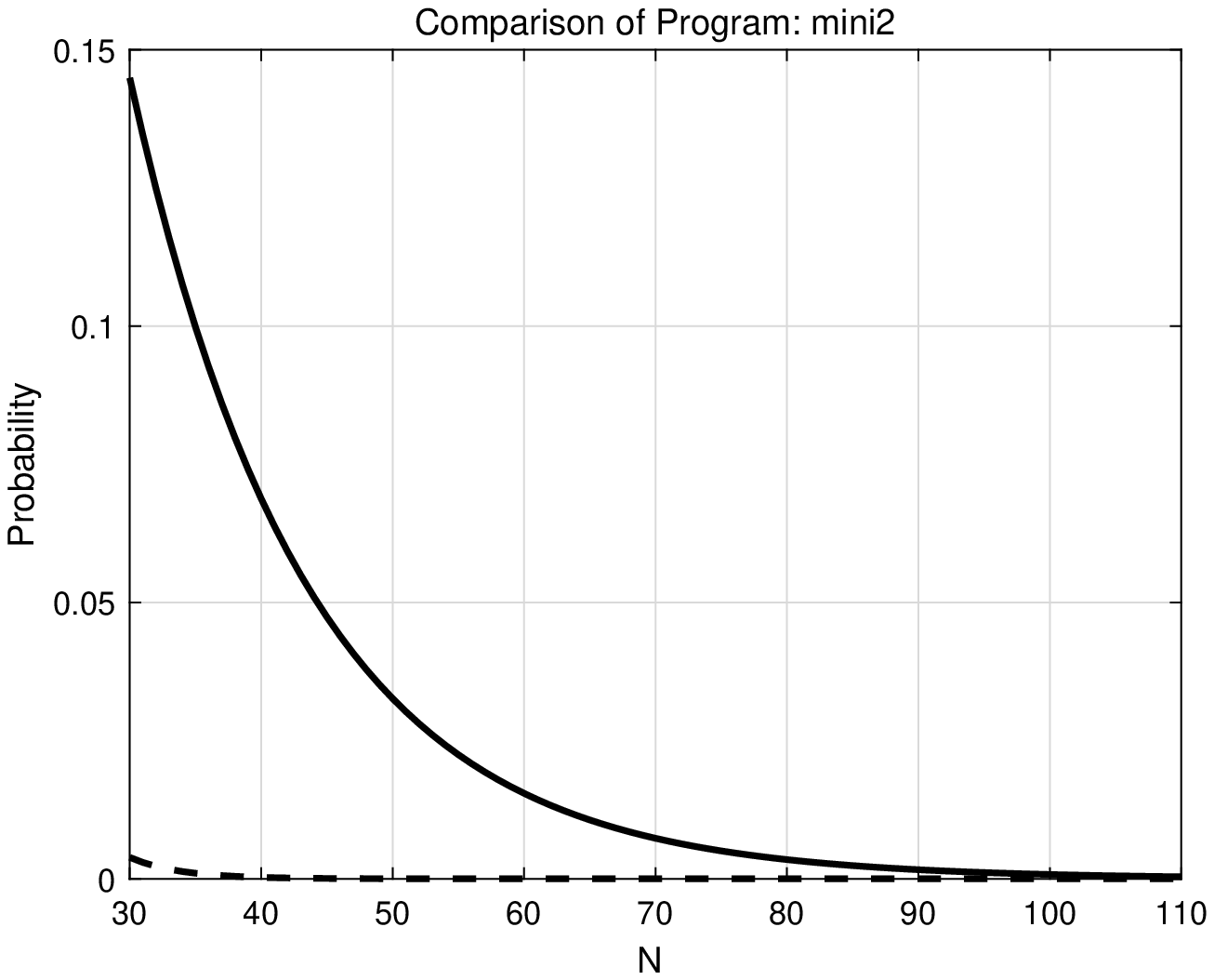} 
		\caption*{(f)This is the result of example mini-roulette2.} 
		\label{fig:side:b} 
	\end{minipage}
	\caption{They are specific figures of experiment results. In these figures, the solid line represents the Hoeffding Bound and the dotted line represents the bound from our approach. }
	\label{figure:ddd} 
\end{figure} 

\section{Detailed Proofs for Section 5}

\subsection{Proof of Lemma \ref{lemma:form}} \label{proof:form}

It is easy to see that this form is closed under addition, multiplication and derivative. 
We first prove the following lemma:
\begin{lemma}\label{lemma:small}
For $0\leq h(n)\leq f(n)$ generated by $\mathfrak{e}$, suppose $g(n)$ is the most significant term (ignoring coefficients), which means $f(n):=q\cdot g(n) + o(g(n))$. Set $c:=\alpha^{g(n)}$, then for every $u\in \mathbb R$, $u\cdot \alpha^{h(n)}$ could be simplified under Step 3's strategy into a univariate function on $c$ as in Lemma \ref{lemma:form}. 

\end{lemma}

\begin{proof}
Suppose $h(n):= \sum_{i} \omega_i h_i(n)$, where $h_i(n) \in \{n, \ln n, n\ln n\}$. Then we have:
\begin{align*}
	u\cdot \alpha^{h(n)} &= u \cdot c^{\frac{h(n)}{g(n)}} \\
	&= u\cdot \prod_{i} {c^{\omega_i\frac{h_i(n)}{g(n)}}}
\end{align*}

For every $i$, if $h_i(n) = g(n)$, then this term would become $c^{\omega_i}$, otherwise $\frac{h_i(n)}{g(n)}\leq 1$ since $h(n)\leq f(n)$. Then if $\omega_i\cdot u>0$, it would be simplifed into $1$, otherwise it would be simplified into $c$. Since the form is closed under multiplication, we conclude that $u\cdot \alpha^{h(n)}$ could be simplified into the form of Lemma \ref{lemma:form}. 

\end{proof}

We now return to proving Lemma~\ref{lemma:form}. By design, our PRR always has the following form:
$$\time(n) = a(n) + \frac{\gamma}{n}\cdot \left(\sum_{i = \lceil n/2\rceil}^{n-1} \time(i) + \sum_{i = \lfloor n/2\rfloor}^{n-1} \time(i)\right) + \frac{1-\gamma}{n}\cdot {\sum_{i = 0}^{n-1}\time(i)}$$
where $0\leq \gamma\leq 1$. We consider several cases:

\paragraph{Case (i)} $f(n) := q\ln n$ and $q>0$. Recall that $f(n)$ is a over-approximation of $\expv[\time(n)]$, $a(n)$ would be $k_1\ln n + k_2$. First note that $k_1 = 0$, since if $k_1 > 0$, then $\mathbb E[\time(n)] = \Omega(\ln^2 n) > f(n)$. Since $k_1 = 0$, we have $k_2\geq 0$. Hence, the conditions for $\alpha$ would be:
\[
\alpha^{q\cdot \ln n}\geq \alpha^{k_2} \cdot \left(\frac{\gamma}{n}\cdot \left(\sum_{i = \lceil n/2\rceil}^{n-1} \alpha^{q\cdot \ln i} + \sum_{i = \lfloor n/2\rfloor}^{n-1} \alpha^{q\cdot \ln i}\right) + \frac{1-\gamma}{n}\cdot {\sum_{i = 0}^{n-1}\alpha^{q\cdot \ln i}}\right)
\]

By integration and concavity, we have:
\[
\alpha^{q\cdot \ln n}\geq \alpha^{k_2} \cdot \left(\frac{2\gamma}{n}\cdot \left(\frac{n\alpha^{q\ln n}-(n/2)\alpha^{q\ln (n/2)}}{q\ln \alpha + 1}\right) + \frac{1-\gamma}{n}\cdot { \frac{n \alpha^{q\ln n}}{q\ln \alpha + 1}}\right)
\]

We define $c:=\alpha^{\ln n}$, we have:
\[
c^q\geq \alpha^{k_2} \cdot \left(\frac{2\gamma}{n}\cdot \left(\frac{nc^q-(n/2)(\frac{c^q}{\alpha^{q\ln 2}})}{q\ln \alpha + 1}\right) + \frac{1-\gamma}{n}{\cdot \frac{n c^q}{q\ln \alpha + 1}}\right)
\]

By moving everything to left and eliminating the fraction, we get:
\[
n c^q (1 + q\ln \alpha)- \alpha^{k_2} \cdot {2\gamma}\cdot \left(nc^q-(n/2)\cdot \frac{c^q}{\alpha^{q\ln 2}}\right) + (1-\gamma)\cdot (n \cdot c^q) \geq 0
\]

By removing common parts $n$ and $c^q$, we derive:
\[
(1 + q\ln \alpha)- \alpha^{k_2} \cdot {2\gamma}\cdot \left({1-(1/2)\cdot \frac{1}{\alpha^{q\ln 2}}}\right) + (1-\gamma) \geq 0
\]

Since $k_2,\gamma,q$ are constants, this is a univariate inequality for $\alpha$ in the form of Lemma \ref{lemma:form}.

\paragraph{Case (ii)} $f(n) := qn, q>0$. Recall that $f(n)$ is a over-approximation of $\mathbb E[\time(n)]$. $a(n)$ would be $k_1 n + k_2\ln n + k_3$, where $0<k_1\leq q, k_2,k_3\in \mathbb R$. Conditions for $\alpha$ would be:
\[
\alpha^{q\cdot  n}\geq \alpha^{k_1 \cdot n + k_2\cdot \ln n + k_3} \cdot \left(\frac{\gamma}{n}\cdot \left(\sum_{i = \lceil n/2\rceil}^{n-1} \alpha^{q\cdot  i} + \sum_{i = \lfloor n/2\rfloor}^{n-1} \alpha^{q\cdot  i}\right) + \frac{1-\gamma}{n}\cdot {\sum_{i = 0}^{n-1}\alpha^{q\cdot  i}}\right)
\]

By integration and concavity, we have:
\[
\alpha^{q\cdot  n}\geq \alpha^{k_1 \cdot n + k_2\cdot \ln n + k_3} \cdot \left(\frac{2\gamma}{n}\cdot \left(\frac{\alpha^{q\cdot n}-\alpha^{q\cdot (n/2)}}{q\ln \alpha}\right) + \frac{1-\gamma}{n}\cdot {\frac{\alpha^{q\cdot n} - 1}{q\ln \alpha}}\right)
\]

By moving everything to left and eliminating the fraction, we get:
\[
n\cdot q\cdot \ln\alpha \cdot\  \alpha^{q\cdot  n}- \alpha^{k_1 \cdot n + k_2\cdot \ln n + k_3} \cdot \left({2\gamma}\cdot \left({\alpha^{q\cdot n}-\alpha^{q\cdot (n/2)}}\right) + ({1-\gamma})\cdot {\left({\alpha^{q\cdot n} - 1}\right)}\right) \geq 0
\]

We define $c:=\alpha^{n}$. By Lemma \ref{lemma:small}, $\alpha^{k_1 \cdot n + k_2\cdot \ln n + k_3}$, and since $n\ln\alpha = \ln c$ could be simplified into the form we want, thus the whole formula could be simplified into the form we want.

\paragraph{Case (iii)} Otherwise, by design, $\alpha^{f(n)}$ will have no elementary antiderivative, in this case we would partition the summation uniformly into $B$ parts, and use the maximum element of a part to over-approximate each of its elements. In this case, conditions for $\alpha$ would be:
\[
\alpha^{f(n)}\geq \alpha^{a(n)} \cdot \left(\frac{\gamma}{n}\cdot \left(2\sum_{i = \lceil (n-1)/2\rceil}^{n-1} \alpha^{f(i)}\right) + \frac{1-\gamma}{n}\cdot {\sum_{i = 0}^{n-1}\alpha^{f(i)}}\right)
\]

By uniformly dividing into $B$ parts and using the maximum element of a part to over-approximate each of its elements, we derive:
\[
\alpha^{f(n)}\geq \alpha^{a(n)} \cdot \left(\frac{\gamma}{n}\cdot \left(\frac{2n}{B}\cdot \sum_{i = 1}^{B} {\alpha^{f\left(\frac{n(B+i)}{2B}\right)}}\right) + \frac{1-\gamma}{n}\cdot \frac{n}{B} \cdot {\sum_{i = 1}^{B}\alpha^{f
\left(\frac{in}{B}\right)}}\right)
\]

By moving everything to the left, and removing the denominators, we obtain:
\begin{align*}
B\cdot n\cdot \alpha^{f(n)}- \alpha^{a(n)} \cdot \left(2\cdot \gamma\cdot n\cdot \sum_{i = 1}^{B} \alpha^{f\left(\frac{n(B+i)}{2B}\right)} + n\cdot (1-\gamma)\cdot \sum_{i = 1}^{B}\alpha^{f\left(\frac{in}{B}\right)}\right) \geq 0
\end{align*}

We remove the common term $n$, and let $c:=\alpha^{g(n)}$, where $g(n)$ is the most siginificant term in $f(n)$ (ignoring coefficients). By Lemma \ref{lemma:small}, the whole formula could be simplified into the form we want.

\subsection{Proof of Lemma \ref{lemma:PosNeg}} \label{proof:Posneg}

\paragraph{Proof for Case (i)} Define $\psi(\infty) := \lim_{t\to \infty} \psi(t)\in \mathbb R \cup \{-\infty\}$. Set $c^*:=\sup\{x|x\geq 1\land \psi(x)\geq 0\}$, then by $\forall x>c^*, \psi(x)<0$ and continuity of $\psi$, we have $\psi(c^*)\geq 0$. Since $\psi$ is monotonically decreasing, $\forall x< c^*, \psi(x)> \psi(c^*)\geq 0$, so $c^*$ is feasible. 

\paragraph{Proof for Case (ii)} Since $\psi'(c)$ is separable, there exists $c'$ such that $\psi'(c)>0$ on $(1,c')$ and $\psi'(c)<0$ on $(c',\infty)$.
By Newton-Leibnitz formula: $\psi(c') = \psi(1) + \int_{1}^{c'} \psi'(t) \text{d}t$, so $\psi(c') \geq \psi(1) \geq 0$. Since $\psi'(c) < 0$ on $(c', \infty)$, then by similarly to (i), we conclude that $\psi$ is separable.

\paragraph{Proof for Case (iii)} If $\psi/(c^a)$ is separable for some $a\in \mathbb R$, then there exists $c^*$ such that $\psi(c)/c^a>0$ on $(1,c^*)$ and $\psi(c)/c^a<0$ on $(c^*,\infty)$, since $c^a>0$ for $c>1$, we conclude that $\psi(c) > 0$ on $(1,c^*)$ and $\psi(c)<0$ on $(c^*,\infty)$, so $\psi$ is separable.  
	
\section{Supplementary material for probabilistic recurrences}

\subsection{On the concentration bound for Quicksort}

Here we emphasize use exactly the same step in our paper to derive the following bound for all $k_1,k_2>0$:
$$\Pr[\time(n^*)\geq (9+k_1)\cdot n^*\ln n^* + k_2\cdot n]\leq (2.3)^{-k_1\cdot \ln n^*-k_2}$$

\textbf{The steps below until the sentence "Finally, we could prove.." are exactly the same as our paper in Page 14, without changing any word}.

Given $f(n):= 9 \cdot n \cdot \ln n$, we apply Lemma \ref{lemm:recurrence} and obtain the following conditions for $\alpha$:\footnote{We assume $0\ln 0 := 0$.}
$$\forall 1\leq n\leq n^*,~~~\alpha^{9\cdot n \cdot \ln n}\geq \alpha^{n-1} \cdot \frac{1}{n} \cdot \sum_{i = 0}^{n-1}{\alpha^{f(i)+f(n-1-i)}}$$
For $1 \leq n\leq 8$, we can manually verify that it suffices to set $\alpha > 1$. In the sequel, we assume $n \geq 8.$
Note that $\left(i \cdot \ln i + (n-i-1) \cdot \ln(n-i-1)\right)$ is monotonically decreasing on $[1,\lfloor n/2\rfloor]$, and monotonically increasing on $[\lceil n/2\rceil ,n]$. We partition the summation above uniformly into eight parts and use the maximum of each part to overapproximate the sum. This leads to the following upper bound for this summation:
\begin{align*}
\sum_{i = 0}^{n-1}{\alpha^{f(i)}}  & \leq \sum_{j=0}^7 \sum_{i=\lceil j \cdot n / 8 \rceil}^{\lfloor (j+1)\cdot n / 8 \rfloor} \alpha^{f(i)+f(n-1-i)}
\\
&\leq\frac{n}{4} \cdot \left(\alpha^{9\cdot n \cdot \ln n} + \alpha^{9 \cdot (\frac{n}{8} \cdot \ln \frac{n}{8} + \frac{7 \cdot n}{8} \cdot \ln \frac{7 \cdot n}{8})} + \alpha^{9\cdot (\frac{n}{4} \cdot \ln \frac{n}{4} + \frac{3\cdot n}{4}\ln \frac{3 \cdot n}{4})} + \alpha^{9 \cdot (\frac{5 \cdot n}{8}\ln \frac{5 \cdot n}{8} + \frac{3 \cdot n}{8}\ln \frac{3 \cdot n}{8})}\right)
\end{align*}
Plugging in this overapproximation back into the original inequality, we get:
$$\alpha^{9\cdot n \cdot \ln n}\geq \alpha^{n-1} \cdot \frac{1}{4} \cdot \left(\alpha^{9\cdot n \cdot \ln n} + \alpha^{9 \cdot (\frac{n}{8} \cdot \ln \frac{n}{8} + \frac{7 \cdot n}{8} \cdot \ln \frac{7 \cdot n}{8})} + \alpha^{9\cdot (\frac{n}{4} \cdot \ln \frac{n}{4} + \frac{3\cdot n}{4}\ln \frac{3 \cdot n}{4})} + \alpha^{9 \cdot (\frac{5 \cdot n}{8}\ln \frac{5 \cdot n}{8} + \frac{3 \cdot n}{8}\ln \frac{3 \cdot n}{8})}\right)$$
for all $8\leq n\leq n^*$. We define $c:=\alpha^{n}$ to do substitution, and we use the following formula to do strengthening:
\begin{align*}
\alpha^{\beta \cdot \ln \beta + (n-\beta) \cdot \ln (n-\beta)}&\leq \alpha^{\beta \cdot \ln n + (n-\beta) \cdot \ln (n-\beta)}
\\&= \alpha^{n \cdot \ln n} \cdot \alpha^{-(n-\beta) \cdot \ln n + (n-\beta) \cdot \ln (n-\beta)}
\\&= \alpha^{n \cdot \ln n + (n-\beta) \cdot \ln \frac{n-\beta}{n}} = c \cdot c^{\frac{(n-\beta)}{n \cdot \ln n} \cdot \ln \frac{n-\beta}{n}}
\end{align*}
By defining $\beta = \frac n 8, \frac n 4, \frac {3n} 8$ respectively, we obtain:
\begin{align*}
 \forall 8\leq n\leq n^*,~~ &c^{9\ln n}\geq \frac{c}{n}\cdot \frac{n}{4} \cdot \left(c^{9\ln n} + c^{9\ln n + \frac{63}{8} \cdot \ln \frac 78} + c^{9\ln n + \frac{27}{4} \cdot \ln \frac 34} + c^{9\ln n + \frac{45}{8} \cdot \ln \frac 58}\right) \\
 \forall 8\leq n\leq n^*,~~ &4 - \left(c + c^{1- \frac{63}{8} \cdot \ln \frac 87} + c^{1 - \frac{27}{4 } \cdot \ln \frac 43} + c^{1 - \frac{45}{8} \cdot \ln \frac 85}\right) \geq 0 \\
\end{align*}
Now we study the following function:
$$\psi(c) = 4 - \left(c + c^{1- \frac{63}{8} \cdot \ln \frac 87} + c^{1 - \frac{27}{4 } \cdot \ln \frac 43} + c^{1 - \frac{45}{8} \cdot \ln \frac 85}\right).$$
By basic calculus, we can prove that $\psi(c)\geq 0$ holds on $[1,2.3]$. Additionally, since for every $\alpha$, $\alpha^{n}$ increases as $n$ increases, by plugging $\alpha = 2.3^{1/(n^*)}$ into Lemma~\ref{lemm:recurrence}, we obtain:
$$\Pr[\time(n^*)\geq (9+k_1)\cdot n^*\ln n^* + k_2\cdot n]\leq (2.3)^{-k_1\cdot \ln n^*-k_2}$$

\end{document}